\documentclass[11pt,letterpaper]{article}
\usepackage[margin=1in]{geometry}
\usepackage{macros}
\usepackage{project-macros}
\usepackage{authblk}

\title{Verifiable Quantum Advantage and Computation\\
via Quantum Circuit Obfuscation}
\author[1]{Alexandru Gheorghiu}
\author[1,2]{Aparna Gupte}
\author[1]{Vojtěch Havlíček}
\author[1]{Yunchao Liu}

\affil[1]{IBM Research}
\affil[2]{Massachusetts Institute of Technology}
\date{\today}

\begin{document}

\sloppy

\maketitle

\begin{abstract}
We construct protocols for classically verifiable quantum advantage and classical verification of $\mathsf{BQP}$ computations using \emph{quantum indistinguishability obfuscation} (qiO). Specifically, given qiO and assuming a slightly stronger version of $\mathsf{BQP}\neq\mathsf{BPP}$, we construct a two-message quantum-advantage protocol that is efficiently and publicly verifiable.
Our result can be viewed as a rigorous cryptographic foundation for the heuristic quantum advantage proposals based on \emph{peaked random circuit sampling} of Aaronson and Zhang (arXiv:2404.14493).

We also construct two simple protocols for classically verifying arbitrary $\mathsf{BQP}$ computations. The first protocol is privately verifiable and assumes only the existence of qiO. This gives a rare example of a nontrivial cryptographic application of (quantum) iO that does not make additional computational hardness assumptions. The second protocol additionally assumes post-quantum one-way functions and is \emph{publicly verifiable}.

As qiO is the basis for our results, we propose an approach to constructing qiO by adapting the local-mixing framework of Canetti--Chamon--Mucciolo--Ruckenstein.
\end{abstract}

\newpage
\tableofcontents
\newpage

\section{Introduction}

Designing a compelling test of near-term quantum advantage has become increasingly important for benchmarking the progress in engineering quantum computers. Beyond requiring such a test to be computationally infeasible for a classical machine, we would ideally like a \emph{proof of quantumness} to satisfy two further properties: (1) it should be implementable on near-term quantum hardware, and (2) its outcome should be efficiently verifiable by a classical auditor.

Unfortunately, existing approaches do not satisfy both requirements simultaneously. Cryptographic proofs of quantumness, such as \cite{BrakerskiEtAl2018Test,kahanamoku2022classically,yamakawa2022verifiable}, allow for efficient classical verification, but are not implementable on present-day noisy devices. Likewise, despite recent advances in optimizing Shor's algorithm \cite{chevignard2024reducing,gidney2025factor}, which make it increasingly plausible for early fault-tolerant quantum computers, it remains out of reach for near-term hardware.

Random circuit sampling~\cite{AruteEtAl2019Supremacy}, on the other hand, is implementable on near-term devices and enjoys substantial complexity-theoretic evidence for classical hardness \cite{boixo2018characterizing,bouland2019complexity,movassagh2019quantum}. However, verifying that the device is sampling from the correct distribution requires classical computation that scales exponentially with the number of qubits~\cite{HangleiterEtAl2019SampleComplexity}.

As an attempt to resolve this tension, Aaronson and Zhang proposed \emph{peaked circuit sampling} \cite{AaronsonZhang2024Peaked}. A peaked circuit is an otherwise complicated quantum circuit whose output distribution places noticeable probability on a hidden computational-basis string. This secret string, called the peak, acts as a built-in verification \emph{trapdoor}. An auditor generates a circuit together with knowledge of its peak and sends the circuit to a purported quantum device. The device repeatedly executes the circuit and measures the output in the computational basis. If the claimed peak appears with the prescribed frequency, the auditor can check the answer essentially for free. However, this proposal raises two immediate questions. First, how can the auditor efficiently generate a sufficiently complicated peaked circuit while retaining knowledge of its peak? Second, given such an efficient generation procedure, how would we justify that this test is not classically spoofable?

A natural way to hide this structure is through \emph{program obfuscation}---a procedure that is meant to ``scramble'' the description of a program (or in our case, a quantum circuit) while preserving functionality. Indeed, a number of works have attempted to engineer large peaked circuits starting from circuits whose peak was easy to engineer and then applying heuristic circuit transformations intended to conceal this structure \cite{ShepherdBremner2009Instantaneous, gharibyan2025heuristic,YungCheng2020AntiForging,BremnerChengJi2025Stabilizer,Yan2026CliffordObfuscation}. Subsequent work was however able, in most cases, to classically recover the peak efficiently \cite{kremer2026efficient,KahanamokuMeyer2023Forging,GrossHangleiter2025SecretExtraction}.

This pattern of candidate constructions followed by subsequent attacks is standard in the field of cryptography. Cryptographers have developed a modular approach to addressing this problem---rather than analyzing the security of each candidate construction from scratch, one seeks to reduce its security to a precisely stated computational assumption, or to another cryptographic primitive that has been independently studied in different contexts. This practice of proving a reduction from a protocol's security to well-studied, widely-believed assumptions instills more confidence in the protocol's security.

In the context of classical program obfuscation, a notion called ``indistinguishability obfuscation'' (iO) has proved particularly powerful in this regard. Beginning with the candidate constructions of Garg et al.\ \cite{garg2013candidate}, iO has been shown to imply a remarkable range of cryptographic primitives and has emerged as a ``central hub'' for cryptographic constructions \cite{sahai2014use}. This philosophy is the motivating principle of this work. We therefore ask

\begin{center}
\emph{Can we base the security of peaked-circuit quantum advantage on quantum indistinguishability obfuscation?}
\end{center}

We answer this question affirmatively. Assuming quantum indistinguishability obfuscation (qiO), and additionally (a slightly stronger variant of) $\BQP \neq \BPP$, we give a proof of quantumness inspired by peaked circuit sampling, that is provably secure against classical spoofing. 
The same ideas can also be used to build a simple protocol for \emph{classical verification of quantum computation} (CVQC), assuming only qiO. Classical verification has historically required sophisticated cryptographic machinery, beginning with Mahadev's breakthrough protocol based on Learning with Errors \cite{mahadev2018classical}. Moreover, making verification \emph{public} and \emph{blind} is subtle in the quantum setting and has been challenging to obtain in the standard model \cite{bartusek2023obfuscation}. Using the same ideas underlying our peaked-circuit construction, we obtain a simple CVQC protocol assuming qiO and show how to compile it, using one-way functions, into a publicly verifiable one. 

These results shift a substantial burden onto construing qiO. General-purpose quantum indistinguishability obfuscation is an extremely strong primitive, and constructing it from standard assumptions remains open \cite{broadbent2020constructions,bartusek2023obfuscation}. We therefore conclude by proposing a plausible concrete realization, adapting the recent local-mixing framework of Canetti et al.\ \cite{canetti2024localmixing} to quantum circuits. They considered a method for obfuscating classical circuits by repeatedly applying local, functionality-preserving transformations to a circuit, with security decomposed into a relatively weak obfuscation property for random circuits and a pseudorandomness assumption. We formulate a quantum analogue of this framework and identify corresponding properties which, if instantiated, would yield quantum indistinguishability obfuscation. We do not provide such an instantiation. Rather, our goal is to separate the cryptographic applications of qiO developed in this work from the concrete problem of constructing it, and to provide a framework in which candidate constructions can be studied and attacked.

\section{Overview of results}

Our main result is to use quantum circuit obfuscation, together with additional computational assumptions in some cases, to construct two cryptographic protocols central to delegated quantum computation: proofs of quantumness and publicly verifiable classical verification of $\BQP$.
Unfortunately, constructing quantum circuit obfuscation from standard cryptographic assumptions has proved to be a challenging endeavor. We therefore also adapt the \cite{canetti2024localmixing} framework, of obfuscating classical circuits via local-mixing, to the quantum setting, with the goal of providing a foundation for heuristic constructions of quantum iO and, more optimistically, constructions that may eventually admit efficient implementations.

This section is structured as follows. In \Cref{sec:to-qio}, we introduce the notion of quantum circuit obfuscation we use in our work---quantum indistinguishability obfuscation where the obfuscating algorithm is itself classical. Readers familiar with indistinguishability obfuscation may skip \Cref{sec:to-qio}. \Cref{sec:to-poq,sec:to-cvqc} describe our proof of quantumness and $\BQP$ verification protocols, and \Cref{sec:to-ccmr} gives an overview of our quantum adaptation of the local-mixing framework of \cite{canetti2024localmixing}.

\subsection{Quantum circuit obfuscation}
\label{sec:to-qio}
Quantum obfuscation can be defined in several models, depending on whether the input, output, and evaluation procedure are classical or quantum. We focus on the model most relevant to delegated quantum computation, where a classical client interacts with an untrusted server that may be classical or quantum. In our applications, the role of quantum circuit obfuscation is to allow the classical client to enforce guarantees on the behavior of this potentially malicious server. It is therefore important that the obfuscation procedure itself can be carried out classically. Accordingly, we consider the setting in which the \emph{obfuscator is a classical algorithm}: it takes as input a classical description of a quantum circuit and outputs a classical description of an obfuscated quantum circuit.

We now discuss the formal definition of obfuscation that we use in this paper. Intuitively, an obfuscation scheme should scramble the description of a circuit, while keeping its functionality intact. It is at first not clear how to turn this into a rigorous and meaningful definition, but classical cryptography has converged on the notion of \emph{indistinguishability obfuscation} (iO) as the ``right'' definition of obfuscation~\cite{barak2001possibility}. An obfuscator is said to satisfy classical iO security if for every two classical circuits $C_1, C_2$ that compute same function (i.e., have the same truth table), their obfuscations should look computationally indistinguishable,
\begin{align*}
    \mathsf{iO}(C_1) \approx_c \mathsf{iO}(C_2).
\end{align*}
That is, given the descriptions of equivalent circuits $C_1$ and $C_2$ of the same size, and an obfuscation $\mathsf{iO}(C_b)$ for a random $b \gets \{1,2\}$, no efficient algorithm should be able to guess $b$ with greater than negligible advantage. Intuitively, this is insisting that the obfuscated circuit must depend (in a computational sense) only on the truth table, not on any particular implementation as a circuit. 

To see why this could be useful, suppose we want to obfuscate a circuit $C$ because its description reveals some secret $s$. For example, coming back to our setting, think of $C$ as a peaked circuit and $s$ as the secret peak. If there exists some circuit $C'$ with the same input/output behavior as $C$, yet $C'$ hides $s$, then so must its obfuscation $\mathsf{iO}(C')$. Since $\mathsf{iO}(C)$ and $\mathsf{iO}(C')$ are computationally indistinguishable, this means that $\mathsf{iO}(C)$ must also hide the secret peak $s$. In this sense, iO achieves the ``best-possible'' security~\cite{goldwasser2007best}.

Remarkably, iO has proved to be a powerful notion---not only can it be constructed from well-founded cryptographic assumptions~\cite{jain2021indistinguishability}, but it also serves as a ``central hub'' for cryptography, enabling a wide range of cryptographic primitives under mild additional assumptions~\cite{sahai2014use,garg2014two,bitansky2016perfect}.

Classical iO has a natural quantum analog---a quantum indistinguishability obfuscation (qiO) scheme must output quantum circuits that are computationally indistinguishable for every two quantum circuits $Q_1, Q_2$ of the same size that are close in diamond distance. That is, for every two quantum circuits $Q_1, Q_2$ such that $\|Q_1(\cdot) - Q_2(\cdot)\|_{\diamond}$ is negligible,\footnote{For a circuit $C$, we write $C(\cdot)$ to denote the channel implemented by the circuit.} we must have that
\begin{align*}
    \qio(Q_1) \approx_c \qio(Q_2).
\end{align*}
This is the notion of quantum circuit obfuscation we study in this work.

\subsection{Proof of quantumness}\label{sec:to-poq}

Our first application of qiO is a proof of quantumness protocol. Here, the goal is to specify an interactive protocol between a classical verifier (or client) and an allegedly quantum prover (or server), that allows a truly quantum server to convince the classical verifier of its quantumness. That is, a quantum prover interacting with the classical verifier can make the verifier accept, while no classical prover should be able to do so. We use quantum iO to build such a protocol that is secure under a very mild additional assumption that seems necessary for quantum advantage---the assumption that $\BQP \neq \BPP$.\footnote{Note that $\BQP \neq \BPP$ on its own does not yield a proof of quantumness; it merely says that there exists some problem that's efficiently solvable on a quantum computer and hard for all efficient classical algorithms. However, this problem need not be verifiable. Further, perhaps surprisingly, proofs of quantumness can exist even if $\BQP = \BPP$. See \cite{mahadev2018classical,TomerZhandry2025Foundations}.}\footnote{We note that we need a slightly stronger assumption than $\BQP \neq \BPP$ for technical reasons. See \Cref{def:iobpp} and the preceding discussion.}

\begin{theorem}[Informal version of \Cref{thm:qio-implies-poq}]\label{informal:poq}
    Suppose that $\BQP \neq \BPP$. Suppose further that there exists a qiO obfuscation scheme secure against classical probabilistic polynomial-time adversaries. Then, there is a proof of quantumness protocol.
\end{theorem}
\noindent In fact, this protocol has a particularly simple structure:
\begin{enumerate}
    \item The verifier samples a bit $b \gets \{0,1\}$ uniformly, computes the obfuscated circuit $\hat{C} \gets \qio(\sfX^b)$ and sends $\hat{C}$ to the prover.\footnote{Here we abuse notation and write $\sfX$ to denote an appropriately padded circuit that applies the $\sfX$ gate on its first qubit.}
    \item The prover runs $\hat{C}$ on the input $\ket{0^n}$ and measures the first qubit to get a bit $b'$, which it then sends to the verifier.
    \item The verifier accepts if $b' = b$, and rejects otherwise.
\end{enumerate}
We claim that no classical prover should be able to guess the bit $b$ correctly with probability significantly better than $1/2$. One can repeat the above protocol many times in parallel to amplify its soundness.

To see why the above scheme is secure, consider a language $L \in \BQP \setminus \BPP$, and let $C^L$ denote an amplified quantum circuit that decides $L$ with negligible error. For any input $x$, define $C^{L,x}$ to be the circuit that hardcodes the input $x$ by applying $\sfX^x$, computes $C^L$, copies the output bit onto a fresh input register, uncomputes $C^L$ as well as $\sfX^x$. Then, we have that $\| C^{L,x}(\cdot)  - \sfX^{L(x)} (\cdot)\|_{\diamond}$ is negligible, therefore by qiO security, we have that for every input $x$,
\begin{align*}
    \qio(C^{L,x}) \approx_c \qio(\sfX^{L(x)}).
\end{align*}
By the assumption that $\BQP \neq \BPP$, we know that for every efficient classical algorithm $\calA$, there must be some hard instance $x^*$ on which it fails to output $L(x^*)$ with non-negligible probability. 
We can intuitively argue soundness by considering the following cases: 
\begin{enumerate}
    \item \textbf{Case 1:} If $x^* \in L$, then $\qio(C^{L,x^*}) \approx_c \qio(\sfX)$.
    \item \textbf{Case 2:} If $x^* \notin L$, then $\qio(C^{L,x^*}) \approx_c \qio(\sfI)$.
\end{enumerate}
Since $\calA$ cannot distinguish between the two above worlds, it must also not be able to distinguish between $\qio(\sfX)$ and $\qio(\sfI)$. This concludes the soundness proof. Note that we only needed that the qiO scheme is secure against \emph{classical} adversaries.

\paragraph{Public verification.}
The above protocol has the downside of being only \emph{privately verifiable}, which means that verification requires the classical client to hold onto a \emph{private state}, namely $b$, in order to verify the prover's response. We would like to make it so that anyone can verify the prover's quantumness, by merely checking the transcript between the prover and the classical client. In particular, the prover itself would be able to tell if the verifier found its response acceptable. This stronger primitive is called a \emph{publicly verifiable} proof of quantumness.
Our protocol above is clearly not publicly verifiable as stated---if the prover's response was just a single bit $b'$, and it knew whether the verifier was going to accept a particular response, a cheating prover could simply try all possibilities and retain the one that would have been accepted.

It turns out that we can upgrade our privately verifiable proof of quantumness into a publicly verifiable one, essentially for free. Consider a version of the protocol that's parallel-repeted $t$ times with secret bits $\vecb = (b_1, \ldots, b_t)$. We can have the classical client publish the image $y$ of $\vecb$ under a (classically secure) public one-way function $\mathsf{OWF}$. Note that it should be hard for a cheating prover to find $\vecb$ given $y$, by the security of the one-way function. In essence, the prover is still forced to evaluate the quantum circuits it received from the verifier.
Once the prover returns $\vecb',$ a third-party can then check whether $\mathsf{OWF}(\vecb') = y$. A hybrid argument proves soundness of this upgraded protocol. See \Cref{thm:public-poq}. Interestingly, the existence of classically secure one-way functions is already implied by the assumptions of \Cref{informal:poq}, so we do not need any additional assumptions for public verifiability in this setting. See \Cref{lem:qio-implies-classical-owf}.

\subsection{Classical Verification of $\BQP$ computation}\label{sec:to-cvqc}

Our second application is a protocol for classically verifying a $\BQP$ computation. Here, a classical client wants to delegate a quantum computation, via some (potentially interactive) protocol, to an untrusted quantum server, while having the guarantee that the output it receives from participating in such a protocol is actually correct. Specifically, we consider the task of verifying $\BQP$ decision problems, where the classical client wants to learn the output resulting from measuring the first qubit of some circuit $C$ on the input $\ket{0^n}$, under the guarantee that this bit is either $0$ or $1$ with probability at least $2/3$ (and can therefore be amplified to all but negligible probability).\footnote{Usually, verification of $\BQP$ is formulated using a fixed $\BQP$ circuit $C$ that decides some language, where the client's goal is to learn whether an input $x$ is in the language or not. Our formulation is equivalent, and this can be seen by hardcoding the input $x$ into the circuit $C$.} We obtain the following theorem.

\begin{theorem}[Informal version of \Cref{thm:private-bqp-verification}]
    Suppose that there exists a qiO obfuscation scheme that is secure against quantum polynomial-time adversaries. Then, there is a (privately-verifiable) protocol for the classical verification of $\BQP$.
\end{theorem}

Similar to our proof of quantumness, the verification protocol also has a simple structure. It hinges on the fact that the following two distributions are computationally indistinguishable, assuming qiO:
\begin{align*}
    D_1 &= \{ \qio(\sfX^b) : b \gets \{0,1\} \}, \text{ and}\\
    D_2 &= \{ \qio(\sfX^b \circ C) : b \gets \{0,1\} \},
\end{align*}
for any amplified $\BQP$ circuit $C$. This holds because each distribution is a balanced mixture of $\qio$ obfuscations of quantum circuits that are equivalent to $\sfX$, and $\qio$ obfuscations of circuits that are equivalent to $\sfI$. 

This observation enables the following strategy: a classical client could send a stream of obfuscated circuits, some of them sampled from $D_1$ (``test'' circuits), and some of them sampled from $D_2$ (``planted'' circuits), and request the quantum server to evaluate all of these circuits. The quantum server cannot detect the locations of the planted circuits, and has no choice but to answer all of them correctly, otherwise it risks being caught by the classical client who knows the correct answers on the test circuits. We can ensure negligible soundness error by having the client accept only if the answers from all the $D_2$ circuits are consistent with a single output value for $C$.

\begin{remark}
    Since $C$ does not need to be explictly given to the prover, this protocol is actually a (computationally) \emph{blind} verification protocol for $\BQP$. This means that the prover only learns an upper bound on the size of $C$ and nothing else.
\end{remark}

\begin{remark}
    Our protocol for private verification of $\BQP$ resembles early trap-based blind verification protocols between a quantum server and a weaker \emph{quantum} client such as that of Fitzsimons--Kashefi~\cite{FitzsimonsKashefi2017Verification}. Their verifier prepares and sends single-qubit states to the prover, hiding ``trap'' qubits (analogous to our test circuits) corresponding to predictable outcomes within the delegated computation. 
    
    The resemblence becomes even stronger with the protocol of Ferracin, Kapourniotis, and Datta~\cite{FerracinKapourniotisDatta2018Verification}, which uses entire trap computations, and the protocol of Leichtle, Music, Kashefi, and Ollivier~\cite{LeichtleEtAl2021Verification}, which interleaves computation and test rounds. All of these protocols use some form of ``blindness'' to hide which parts perform the computation and which test the prover, and achieve information-theoretic security in the setting where the verifier is capable of preparing single-qubit states. 

    Our protocol uses a similar hiding principle with a \emph{fully classical verifier}; its blindness guarantee is \emph{computational} and follows from qiO.
\end{remark}

\begin{remark}
    Interestingly, our protocol relies entirely on qiO security, with no additional computational assumptions. In the classical setting, it is atypical for constructions of cryptographic primitives from iO to depend solely on iO security; usually, they require additional (albeit minimal) computational hardness assumptions, such as the existence of one-way functions. This is because iO exists if $\mathsf{P} = \mathsf{NP}$~\cite{barak2001possibility}, a world in which most useful cryptography cannot exist. The few examples of cryptographic primitives that are implied solely by iO, such as witness encryption and statistical zero-knowledge arguments for $\mathsf{NP}$, become meaningless if $\mathsf{P} = \mathsf{NP}$.

    In our setting, it is possible for qiO to exist even if $\BQP = \QMA$.\footnote{If $\BQP = \mathsf{QMA}$, a quantum algorithm can find the smallest circuit equivalent to a given circuit and use it as the obfuscation. This gives a quantum obfuscator, not necessarily a classical one, which is what the relevant variant of qiO requires.} In such a world, a classical verifier would still be able to verify computations of a $\BQP$ machine, even though the $\BQP$ machine has all the power of $\QMA$. In such a world, we would have that $\mathsf{P} \subseteq \NP \subseteq \BQP = \QMA$. Further, it is still possible that $\BQP \neq \BPP$, in which case a classical verification protocol for $\BQP$ is still a useful primitive.

    By comparison, the soundness of Mahadev's classical verification protocol~\cite{mahadev2018classical} relies on quantum-secure trapdoor claw-free functions, which in particular implies that $\NP \setminus \BQP$ is non-empty, forcing $\BQP \neq \QMA$.
\end{remark}

\paragraph{Public verification of $\BQP$.}
The verification protocol described above is \emph{privately verifiable}, since the verifier needs to know the correct answers on all the test circuits, a fact which needs to remain hidden from the prover. If the $i$'th circuit is an obfuscation of $\sfX^{b_i} \circ C^{\mathbbm{1}[i \in J^*]}$, where $b_i\gets \{0,1\}$ and $J^*$ is the set of indices corresponding to the planted locations, there are exactly two responses that the verifier will accept, namely $\vecb$, and $\vecb \oplus \mathbbm{1}[J^*]$. A natural idea would be to publish the images of both these valid responses, i.e., publish $\mathsf{OWF}(\vecb)$ and $\mathsf{OWF}(\vecb \oplus \mathbbm{1}[J^*])$. We analyze a slight variant of this protocol and prove that it is publicly verifiable and sound in \Cref{sec:public-bqp-verification}. For this protocol, we additionally need quantum-secure one-way functions, which exist as long as qiO exists, and $\NP \setminus \BQP$ is non-empty.

\begin{theorem}[Informal version of \Cref{thm:public-bqp-verification}]
    Suppose that there exists a qiO obfuscation scheme that is secure against quantum polynomial-time adversaries. Additionally, suppose that quantum-secure OWFs exist. Then, there is a publicly-verifiable protocol for the classical verification of $\BQP$.
\end{theorem}

\begin{remark}[Comparison with Bartusek--Malavolta]
\label{rem:comparison-bm22}
Bartusek and Malavolta~\cite{BartusekMalavolta2022NullObfuscation}
construct publicly verifiable NIZK arguments for $\mathsf{QMA}$,
and hence CVQC, using witness encryption, which they obtain from quantum
null-iO. Our protocol can be seen as a conceptually simpler protocol of CVQC, achieving slightly different guarantees---\cite{BartusekMalavolta2022NullObfuscation} use just quantum null-iO, where as we use quantum iO for BQP circuits, and our protocol additionally achieves blindness.
\end{remark}

\subsection{Towards constructing quantum iO}
\label{sec:to-ccmr}

Canetti, Chamon, Mucciolo, and Ruckenstein~\cite{Canetti2024Towards} proposed an approach to classical obfuscation based on local mixing in the space of equivalent reversible circuits. Their approach can be viewed as a \emph{worst-case-to-average-case reduction} for program obfuscation under a pseudorandomness assumption called \emph{split-circuit pseudorandomness} (SCP). 
Roughly, SCP asserts that a uniformly random implementation of $C_1|C_2$ is computationally indistinguishable from the concatenation of sufficiently long, uniformly random implementations of $C_1|R$ and $R^{-1}|C_2$, where $R$ is a random circuit of a prescribed length. The two implementations are sampled independently given $R$, and their lengths sum to the length of the implementation of $C_1|C_2$.
Under SCP, their transformation takes an obfuscator for random reversible circuits and achieves obfuscation for \emph{all} reversible circuits. It does so by using random-looking identity circuits to obfuscate individual gates, and then ``soldering'' adjacent obfuscated pieces by reobfuscating their boundaries. The resulting obfuscator satisfies a guarantee stronger than iO: its output is computationally indistinguishable from a uniformly random equivalent circuit of a prescribed larger size, up to efficient postprocessing. They also propose candidate obfuscators for random circuits based on iterated local, functionality-preserving changes.

We adapt this framework to the quantum setting, using stronger, simpler-to-state variants of their assumptions against quantum distinguishers. These yield obfuscation for ancilla-free circuits implementing \emph{exactly} the same unitary (up to global phase). An additional \emph{stability assumption} extends this guarantee to circuits whose induced unitaries are negligibly close (\Cref{cor:ancilla-free-qio}). 

Our applications additionally require obfuscating circuits whose equivalence is guaranteed only when their ancillas are initialized to zero, and whose behavior on dirty-ancilla inputs may differ arbitrarily. A modification of a lemma by \cite{canetti2024localmixing}, shows that qiO for ancilla-free circuits implies qiO for the circuits we need to obfuscate, which means that all our results follow from this weaker notion of qiO for ancilla-free circuits. We describe this transform below.

Canetti et al.~\cite{Canetti2024Towards} encounter an analogous issue when embedding Boolean circuits into reversible ones: two Boolean circuits with the same truth table may have reversible implementations that behave differently on incorrectly initialized workspaces. In both these settings, an adversary could evaluate the obfuscated circuit on such inputs to learn information about its implementation beyond its intended functionality. \cite{canetti2024localmixing} resolve this by constructing a reversible circuit that implements the permutation
\[
(x,y,z,u)\longmapsto
\begin{cases}
(x,y\oplus C(x),z,u), & z=0,\\
(x,y,z,u), & z\neq 0,
\end{cases}
\]
where $z$ is the workspace and $u$ consists of helper wires that need not be initialized. Thus the entire permutation depends only on the truth table of $C$. We adapt this embedding to \emph{tidy} quantum computations, which implement a logical unitary and restore their zero-initialized ancillas to zero, up to negligible error, on every logical input. Writing $V_C$ for the full unitary implemented by $C$, our compilation is
\[
\Gamma(V_C)=K(V_C)F K(V_C)^\dagger F,
\]
where $K(V_C)$ applies $V_C$ controlled on an additional qubit being zero, and $F$ flips that qubit exactly when the workspace is zero. 

For an exactly tidy circuit with logical action $U$, this implements $\operatorname{diag}(U,U^\dagger)$ on the clean-workspace subspace and the identity on its orthogonal complement, removing dependence on the original dirty-workspace behavior. However, a global phase of $U$ becomes a relative phase between the two control branches, so equivalent logical channels need not yield equivalent compiled channels. We address this by sampling $\alpha$ uniformly from a sufficiently fine finite phase grid and instead obfuscating $\Gamma(e^{i\alpha}V_C)$, after negligible-error gate synthesis and padding. Initializing the control and workspace to zero recovers the desired logical computation. This yields qiO for tidy computations from robust ancilla-free qiO without further cryptographic assumptions (\Cref{thm:ancilla-free-to-tidy-qio}).

We stress that we do not instantiate the local-mixing assumptions; designing secure local-rewrite constructions suitable for near-term quantum devices is an interesting direction for future work.

\subsection{Related work}

\paragraph{Quantum indistinguishability obfuscation.}
The framework of Alagic and
Fefferman~\cite{AlagicFefferman2016QuantumObfuscation} allows quantum
obfuscation to output either a circuit description or a quantum state,
and distinguishes perfect, statistical, and computational security.
We use its computational notion with an efficient classical obfuscator
that outputs a classical description of a quantum circuit. This choice
lets our classical verifier generate the obfuscations and lets the prover
copy and evaluate them repeatedly. Security compares same-shape circuits
with negligibly close induced channels. Our verification protocols need
security against quantum adversaries; the proof of quantumness needs
only classical security.

We remark that several recent works that construct obfuscation for quantum circuits rigorously from standard cryptographic assumptions work only in an idealized model, where we assume the existence of \emph{ideal obfuscation} of classical circuits. The ideal obfuscation of a circuit can be thought of as black-box oracle access to the circuit. This is the strongest form of obfuscation, stating that the obfuscation leaks no more information about the circuit other than its input/output behavior.
Although ideal obfuscation for classical circuits is impossible in general~\cite{barak2001possibility}, it has become a useful tool to study the relationship between obfuscation for classical and quantum circuits.
In the classical oracle model, quantum hardness of LWE
enables reusable quantum-state obfuscations of pseudo-deterministic
circuits~\cite{BartusekEtAl2023PseudodeterministicObfuscation}.
Extensions allow the input program to contain a quantum
state~\cite{BartusekEtAl2024QuantumStateObfuscation}. Other schemes
handle approximately unitary programs~\cite{HuangTang2025UnitaryObfuscation} and arbitrary
quantum circuits~\cite{HuangTang2026Obfuscation}. A separate line develops
quantum obfuscation and its applications using quantum
oracles~\cite{ColadangeloGunn2024QuantumObfuscation}. These constructions
allow quantum program states; even a classical oracle, which specifies
a classical function, may permit superposition queries. Closer to our
representation requirements, the construction of Bartusek, Gupte,
Mutreja, and Shmueli~\cite{BartusekEtAl2025ClassicalObfuscation} uses a
classical obfuscator and classical output, assuming quantum-hard LWE
in the classical oracle model. It preserves the classical-input/output
functionality of pseudo-deterministic quantum circuits. Our applications
require channel-preserving qiO without idealized oracles. Practical tools
such as ObfusQate~\cite{BartakeEtAl2025ObfusQate} address a different goal:
they transform code and circuits to resist reverse engineering, without
establishing cryptographic indistinguishability.

\paragraph{Proofs of quantumness and classical verification.}
Trapdoor claw-free functions, instantiated from LWE, underpin the
interactive test of Brakerski et al.~\cite{BrakerskiEtAl2018Test}.
Tests in the quantum random-oracle model reduce interaction to two
messages~\cite{simplerTests,yamakawaZhandry}. Knowledge assumptions offer
another way to achieve a single round from DDH or LWE~\cite{arabadjeiva};
stronger assumptions also permit very shallow quantum
provers~\cite{gheorghiu}. Our two-message construction draws its hardness
from qiO and the worst-case separation
$\BQP\not\subseteq\mathrm{i.o.}\BPP$, without relying on a specific
algebraic or lattice problem.

The distinction between assumptions also matters for verification of
general $\BQP$ computations. Mahadev's
protocol~\cite{Mahadev2018ClassicalVerification} uses quantum hardness
of LWE to let a classical verifier certify the outcome of a quantum
computation. Combining LWE with a quantum random oracle allows
non-interactive verification after setup, although the verifier still
needs secret setup information~\cite{AlagicEtAl2020Noninteractive}.
Our private-verification protocol uses qiO alone. Its hiding mechanism
also connects to trap-based blind verification: a verifier who can
prepare single-qubit states can conceal traps within a
computation~\cite{FitzsimonsKashefi2017Verification} or interleave entire
computation and test
rounds~\cite{FerracinKapourniotisDatta2018Verification,LeichtleEtAl2021Verification}.
Those protocols achieve information-theoretic security through limited
quantum communication; qiO supplies computational hiding for our fully
classical verifier.

Ilango and Lombardi~\cite{IlangoLombardi2025} also use obfuscation as a worst-case-to-average-case reduction to construct a proof of quantumness: assuming subexponentially secure classical iO and one-way functions, they obtain a proof based on the worst-case classical hardness of the white-box Simon problem~\cite{IlangoLombardi2025}. In contrast, our construction uses qiO directly and bases classical soundness on the more general complexity-theoretic assumption $\BQP\not\subseteq\mathrm{i.o.}\BPP$.

\paragraph{Public verifiability.}
Publicly verifiable protocols let anyone check a response from the
public transcript. Quantum null-iO with a classical
obfuscator offers one route: witness encryption and additional classical
cryptographic tools turn it into publicly verifiable arguments for
$\mathsf{QMA}$, and hence $\BQP$~\cite{BartusekMalavolta2022NullObfuscation}.
That reduction treats quantum null-iO as a premise; its proposed
instantiations use a quantum random oracle or an ideal classical oracle.
Publicly verifiable quantum fully homomorphic encryption offers another
route to public and blind classical verification in the classical oracle
model~\cite{BartusekEtAl2025ClassicalObfuscation}. Our protocol uses qiO
for $\calC_{\BQP}$ and post-quantum one-way functions directly: an
honestly generated, circuit-dependent classical challenge suffices to
check the prover's classical response. We also give a route to
instantiating the required qiO without idealized oracles, under quantum
local-mixing and stability assumptions
(\Cref{cor:ancilla-free-qio,thm:ancilla-free-to-tidy-qio}).
Communication and verification cost may grow polynomially with the
delegated circuit's size. For proofs of quantumness, the random-oracle
construction of Yamakawa and Zhandry~\cite{yamakawaZhandry} also admits
public checks; ours derives an efficient public verifier from qiO and
the stated complexity separation.

\paragraph{Local mixing.}
Local mixing aims to hide a circuit's structure through random
perturbations that preserve its functionality. For classical reversible
circuits, the framework of Canetti, Chamon, Mucciolo, and
Ruckenstein~\cite{Canetti2024Towards} reduces obfuscation of arbitrary
circuits to obfuscation of random bounded-length circuits under
split-circuit pseudorandomness. The reduction establishes the stronger
random-output indistinguishability guarantee. We carry its random masks,
inverse circuits, and joining argument over to quantum unitaries, with
assumptions that account for quantum distinguishers and correlated side
information. These assumptions concern the circuit descriptions
themselves, beyond black-box access to random unitaries. Our reduction
first handles exactly equivalent ancilla-free circuits; stability extends
it to negligibly close channels. Adapting the reversible embedding to
tidy quantum computations then requires us to handle the clean-ancilla
subspace and global phases. Together, these steps give a conditional
route to the qiO needed by our protocols. Establishing the assumptions
for concrete local rewrites remains open.

\paragraph{Peaked circuit sampling.}
A hidden, high-probability output string can serve as a candidate
certificate of quantum advantage. This motivates peaked circuit sampling,
which seeks to conceal such a string in a circuit that appears
random~\cite{AaronsonZhang2024Peaked}. Hidden Code
Sampling~\cite{DeshpandeEtAl2025Peaked} pursues a related approach with
conditionally peaked outputs. Under its assumptions on simulation cost,
verification is cheaper than full simulation, although fully efficient
verification remains open. Our construction uses obfuscation to realize
this hiding principle with polynomial-time public verification and
negligible classical soundness error. These guarantees follow from our
cryptographic and complexity assumptions; they do not establish security
for either sampling ensemble.

\paragraph{Relation to recent cryptographic characterizations.}
The existence of proofs of quantumness with an inefficient verifier is
equivalent to the existence of classically secure one-way
puzzles~\cite{MorimaeShirakawaYamakawa2025Characterization}. Our
one-way-function consequence supplies this weaker puzzle primitive,
while our protocol additionally provides an efficient public verifier.
A complementary barrier rules out fully black-box constant-round
constructions of such proofs from post-quantum iO for classical circuits
and one-way permutations~\cite{TomerZhandry2025Foundations}. Our use of
quantum-circuit obfuscation and a classical hardness separation for
$\BQP$ falls outside that barrier. Another line connects obfuscation
of \emph{classical} circuits, with quantum evaluation or encoding, to
pseudorandom unitaries and other quantum cryptographic primitives under
$\mathsf{NP}\not\subseteq\mathrm{i.o.}\BQP$
~\cite{MorimaeShirakawaYamakawa2026QiO}. The hardness assumptions serve
different roles here: our proof of quantumness uses classical hardness
of $\BQP$, private verification needs no additional separation, and
public verification adds post-quantum one-way functions.

\section{Preliminaries}

\paragraph{Notation.} We abbreviate ``probabilistic polynomial-time'' and ``quantum polynomial-time'' as \ppt and \qpt respectively.

\subsection{Quantum Circuits}
\label{sec:quantum-circuits}
Throughout this paper, we think of quantum circuits as being specified by a classical description, in the form of a sequence of unitary and measurement gates. We do not consider general quantum programs that require a quantum auxiliary input state to evaluate. By deferring measurements and purifying discarded randomness, we may regard such a circuit $C$ as a unitary $U_C$ acting on an input register ``$\inp$'' and an ancilla register ``$\anc$'' initialized to $\ket{0^{n_{\anc}}}$.  A designated output register ``$\out$'' is retained and the remaining register ``$\junk$'' is discarded. Thus, on an $n_{\inp}$-qubit state $\rho$, we write $C(\cdot)$ to denote the channel implemented by $C$
\begin{align}
    C(\rho) := \tr_{\junk}\!\left[U_C\bigl(\rho_{\inp}\otimes \ketbra{0^{n_{\anc}}}_{\anc}\bigr)U_C^\dagger \right].
    \label{eq:circuit-channel}
\end{align}
We use the same symbol $C$ for a circuit description and its induced channel when the meaning is clear.

We fix a finite universal gate set, so that the gate count and description length are polynomially related. Further, we pick the gate set so that it contains $\mathsf{I}$, $\sfX$, and $\CNOT$, and is closed under adjoints.
We define the \emph{shape} of a circuit $C$ to be the tuple $\shape(C):=(s,n_{\inp},n_{\anc},n_{\out})$, where $s$ is the gate count. Its \emph{width} is $n_{\inp}+n_{\anc}=n_{\out}+n_{\junk}$. We say that a family of circuits $\{C_{\lambda}\}_{\lambda \in \NN}$ indexed by a security parameter $\lambda$ is polynomial-sized if there is a polynomial $\poly$ such that the gate count and width of $C_{\lambda}$ is bounded by $\poly(\lambda)$.

Fix the negligible class tolerance $\eta_{\mathrm{cls}}(\lambda):=2^{-\lambda}$.
We use the following circuit classes.
\begin{itemize}

    \item \textbf{Unitary circuits.}
    A circuit $C$ is unitary if there is a unitary $U$ on
    $n_{\inp}=n_{\out}$ qubits such that
    \[
        \|C-\mathcal{U}\|_{\diamond}
          \leq\eta_{\mathrm{cls}}(\lambda),
        \qquad
        \mathcal{U}(\rho):=U\rho U^\dagger.
    \]
    We denote this class by $\calC_{\unitary}$.

    \item \textbf{Tidy circuits.} Let $C$ be a circuit with an $n_{\inp}$-qubit logical input register, an $n_{\out}$-qubit logical output register, and an $n_{\anc}$-qubit ancilla register. Let $U_C$ be unitary implements on the entire $n = n_{\inp} + n_{\anc}$ qubits.
    Define an isometry $J$ from $n_{\inp}$ qubits to $n$ qubits defined as $J \ket{\psi} = \ket{\psi} \ket{0^{n_{\anc}}}$. Let $S = \mathrm{im}(J) = \mathsf{I} \otimes \ketbra{0^{n_{\anc}}}$. We say that $C$ is a \emph{tidy computation} of an $n$-qubit unitary $U^*_C$ if
    \[
        \inf_{\theta\in\RR}
        \bigl\|U_C \circ J - e^{i\theta}J \circ U^*_C\bigr\|_{2}
        \leq\negl(\lambda).
    \]
Thus, on every logical input, $C$ implements $U_C$ and returns its ancilla
register to $\ket{0^a}$, up to negligible error.  We write
$\calC_{\unitary}^{\mathsf{tidy}}$ for the class of such circuits.

\item \textbf{Tidy $\BQP$ circuits.} We define the class $\calC_{\BQP}$ to be circuits with $n_{\inp} = n_{\out} = 1$, such that for every $C \in \calC_{\BQP}$ of description length $\lambda$, there is a bit $b_C$ such that $\| C - \sfX^{b_C}(\cdot)\|_{\diamond} \le \negl(\lambda)$.

    Thus these circuits coherently encode a decision bit as either the identity or a bit flip on one logical input qubit; in particular, evaluating the circuit on $|0\rangle$ reveals the decision bit. An amplified $\BQP$ computation on a hardwired instance can be compiled into this form: compute its decision into a workspace qubit, CNOT that qubit into the input wire, and uncompute the workspace.

\item A circuit is \textbf{ancilla-free} if $n_{\anc}=0$.  For a circuit class
$\calC$, we write $\calC^{\ancfree}$ for its ancilla-free subclass.
\end{itemize}

\subsection{Quantum Indistinguishability Obfuscation}
\label{sec:qio}

\begin{definition}[Quantum circuit obfuscator]\label{def:qco}
Let $\calC=\{\calC_\lambda\}_{\lambda\in\NN}$ be a class of quantum circuits. An \emph{obfuscator} for $\calC$ is a \ppt\footnote{One could allow the obfuscator $\Obf$ to be a \qpt algorithm, but in this work, we require a classical verifier to perform the obfuscation, we restrict to classical obfuscators.} algorithm $\Obf$ that, on input $(1^\lambda,C)$ with $C\in\calC_\lambda$, and outputs a circuit $\widehat C\gets\Obf(1^\lambda,C)$. Further, it must it satisfy the following correctness guarantee.\footnote{It turns out that for our applications, a weaker notion of correctness suffices, where functionality preserved only on the $\ket{0}$ input. For simplicity, we state only the stronger, more standard, correctness guarantee.}

\textbf{Correctness:} For every $\lambda \in \NN$ and every $C\in\calC_\lambda$,
\[
    \Pr_{\widehat C\gets\Obf(1^\lambda,C)}
       \!\left[\|\widehat C-C\|_\diamond\leq\negl(\lambda)\right]
    \geq 1-\negl(\lambda).
\]
\end{definition}

We now generalize \emph{classical} indistinguishability obfuscation~\cite{barak2001possibility} to the quantum setting. Classically, circuits are equivalent if they compute the same truth tables. Quantumly, we define circuits $C_1, C_2$ to be equivalent if they implement channels that are negligibly close in diamond distance.

\begin{definition}[Quantum iO (qiO)]
\label{def:qio}
An obfuscator $\Obf$ for $\calC$ satisfies \emph{indistinguishability-obfuscation security} if the following holds. For every two (not necessarily uniformly generated) polynomial-size circuit families $\{C_{0,\lambda}\}_{\lambda\in\NN}$ and $\{C_{1,\lambda}\}_{\lambda\in\NN}$ in $\calC$ such that for every $\lambda \in \NN$,
\[
\shape(C_{0,\lambda}) =\shape(C_{1,\lambda})
\quad\text{and}\quad
\|C_{0,\lambda}-C_{1,\lambda}\|_\diamond\leq \negl(\lambda),
\]
and every (non-uniform) polynomial-time adversary $\calA$,
\begin{align*}
 \Bigl|
   \Pr\!\left[
      \calA(1^\lambda,\Obf(1^\lambda,C_{0,\lambda}))=1
   \right]
   -
   \Pr\!\left[
      \calA(1^\lambda,\Obf(1^\lambda,C_{1,\lambda}))=1
   \right]
 \Bigr|
 \leq \negl(\lambda).
\end{align*}
Depending on whether we quantify $\calA$ over all \ppt or \qpt algorithms, $\Obf$ is secure against classical or quantum adversaries.  \footnote{The proof of quantumness in \Cref{sec:poq} needs only classical
security, whereas the verification application in
\Cref{sec:public-bqp-verification} assumes quantum security.}
\end{definition}

\subsection{Proofs of Quantumness}
\begin{definition}[Two-message proof of quantumness]\label{def:poq}
Let $\ell$ be a polynomial.  A two-message (one-round) proof of quantumness
is a tuple of algorithms $(\Challenge,\Prove,\Verify)$ with the following
syntax.
\begin{itemize}
    \item $\Challenge(1^\lambda;r)\to\chall$ is a classical polynomial-time
    challenge algorithm with coins
    $r\in\{0,1\}^{\ell(\lambda)}$.
    \item $\Prove(1^\lambda,\chall)\to\resp$ is a quantum polynomial-time
    prover that returns a classical response.
    \item $\Verify(1^\lambda,r,\chall,\resp)\to\{0,1\}$ is a classical
    polynomial-time decision algorithm.
\end{itemize}
There is a negligible function $\mu$ such that, for every $\lambda$,
\begin{align*}
 \Pr\!\left[
   \Verify(1^\lambda,r,\chall,\resp)=1
   :
   \begin{array}{l}
     r\gets\{0,1\}^{\ell(\lambda)},\\
     \chall\gets\Challenge(1^\lambda;r),\\
     \resp\gets\Prove(1^\lambda,\chall)
   \end{array}
 \right]\geq 1-\mu(\lambda).
\end{align*}
Moreover, for every classical probabilistic polynomial-time prover $\calA$,
there is a negligible function $\mu_{\calA}$ such that, for every $\lambda$,
\begin{align*}
 \Pr\!\left[
   \Verify(1^\lambda,r,\chall,\resp)=1
   :
   \begin{array}{l}
     r\gets\{0,1\}^{\ell(\lambda)},\\
     \chall\gets\Challenge(1^\lambda;r),\\
     \resp\gets\calA(1^\lambda,\chall)
   \end{array}
 \right]\leq \mu_{\calA}(\lambda).
\end{align*}
The first property is \emph{completeness} and the second is
\emph{classical soundness}.  Since the verifier sends a challenge before
receiving the prover's response, this is a two-message protocol, not a
non-interactive protocol in the plain model.
\end{definition}

\begin{definition}[Publicly verifiable Proof of Quantumness]
\label{def:publicly-verifiable-poq}
A two-message proof of quantumness as in \Cref{def:poq},
\[
    \Pi=(\Verify_1,\Prove,\Verify_2),
\]
is \emph{publicly verifiable} if there is a deterministic classical
polynomial-time algorithm $\mathsf{PubVerify}$ such that, for every $r$,
every honestly generated
$\chall=\Verify_1(1^\lambda;r)$, and every response $\resp$,
\[
    \Verify_2(1^\lambda,r,\chall,\resp)
      =\mathsf{PubVerify}(1^\lambda,\chall,\resp).
\]
Thus, after generating the challenge, the verifier retains no secret state
needed to check a transcript.
\end{definition}

\subsection{Classical Verification of $\BQP$ Computation}

\begin{definition}[Private-verifier preprocessing verification]
\label{def:private-bqp-verification}
A private-verifier preprocessing protocol for $\calC_{\BQP}$ is a tuple
$(\Gen_{\mathrm{priv}},\Prove_{\mathrm{priv}},
\Verify_{\mathrm{priv}})$ with the following syntax.
\begin{itemize}
    \item
    $\Gen_{\mathrm{priv}}(1^\lambda,C)\to
      (\mathsf{ch},\mathsf{st})$
    is a classical probabilistic algorithm running in time polynomial in
    $\lambda$ and $|C|$.  It sends $\mathsf{ch}$ to the prover and keeps
    $\mathsf{st}$ secret.
    \item
    $\Prove_{\mathrm{priv}}(1^\lambda,\mathsf{ch})\to w$
    is a quantum polynomial-time algorithm that returns a classical string.
    \item
    $\Verify_{\mathrm{priv}}
       (1^\lambda,C,\mathsf{ch},\mathsf{st},w)
       \to\{0,1,\bot\}$
    is a deterministic classical algorithm running in time polynomial in
    $\lambda$ and $|C|$.
\end{itemize}
For every polynomial-size circuit ensemble
$\{C_\lambda\}_{\lambda\in\NN}$ in $\calC_{\BQP}$, let $a_\lambda$ be its
semantic bit as in \Cref{def:public-bqp-verification}.  There is a negligible
function $\mu$ such that
\begin{align*}
 \Pr\!\left[
   \Verify_{\mathrm{priv}}(
      1^\lambda,C_\lambda,\mathsf{ch},\mathsf{st},w)
      =a_\lambda
   :
   \begin{array}{l}
   (\mathsf{ch},\mathsf{st})
       \gets\Gen_{\mathrm{priv}}(1^\lambda,C_\lambda),\\
   w\gets\Prove_{\mathrm{priv}}
       (1^\lambda,\mathsf{ch})
   \end{array}
 \right]
 \geq 1-\mu(\lambda).
\end{align*}
Moreover, for every possibly nonuniform quantum polynomial-time prover
$\calA$, there is a negligible function
$\mu_{\calA,\{C_\lambda\}}$ such that
\begin{align*}
 \Pr\!\left[
   \Verify_{\mathrm{priv}}\!\left(
      1^\lambda,C_\lambda,\mathsf{ch},\mathsf{st},
      \calA(1^\lambda,\mathsf{ch})
   \right)=1-a_\lambda
   :
   (\mathsf{ch},\mathsf{st})
       \gets\Gen_{\mathrm{priv}}(1^\lambda,C_\lambda)
 \right]
 \leq\mu_{\calA,\{C_\lambda\}}(\lambda).
\end{align*}
As in \Cref{def:public-bqp-verification}, soundness rules out only the wrong
bit; a malicious prover may cause rejection or the correct output.
\end{definition}

We call such a protocol \emph{computationally blind up to circuit shape}
if there is a \ppt simulator $\mathsf{Sim}$ such that, for every
polynomial-size circuit family $\{C_\lambda\}_{\lambda\in\NN}$ in
$\calC_{\BQP}$,
\[
    \mathsf{ch}_\lambda
       \approx_c\mathsf{Sim}(1^\lambda,\shape(C_\lambda)),
    \qquad
    (\mathsf{ch}_\lambda,\mathsf{st}_\lambda)
       \gets\Gen_{\mathrm{priv}}(1^\lambda,C_\lambda),
\]
where indistinguishability is against possibly nonuniform \qpt
distinguishers. The prover receives no subsequent message revealing the
verifier's decision, so simulating the challenge also simulates the
prover's view. In particular, challenges for any two same-shape circuit
families are computationally indistinguishable, even when their answers
differ.

The setup algorithm below is circuit dependent and honestly generated.  Once
its output has been published, the prover sends one classical proof and any
classical party can verify it without secret state.  Thus the protocol is
non-interactive after preprocessing; it is not a non-interactive protocol in
the plain model.

\begin{definition}[$2$-message Publicly verifiable Classical Verification of Quantum Computation]
\label{def:public-bqp-verification}
A publicly verifiable preprocessing protocol for the coherent decision class
$\calC_{\BQP}$ is a tuple $(\Gen,\Prove,\Verify)$ with the following syntax.
\begin{itemize}
    \item $\Gen(1^\lambda,C)\to\pk$ is a classical probabilistic
    setup algorithm running in time polynomial in $\lambda$ and $|C|$.
    \item $\Prove(1^\lambda,\pk)\to\pi$ is a quantum polynomial-time
    algorithm that outputs a classical proof.
    \item $\Verify(1^\lambda,\pk,\pi)\to\{0,1,\bot\}$ is a deterministic
    classical polynomial-time algorithm.
\end{itemize}
Only the setup algorithm receives $C$; the prover and public verifier use
the public challenge $\pk$.
For every polynomial-size circuit ensemble
$\{C_\lambda\}_{\lambda\in\NN}$ with
$C_\lambda\in\calC_{\BQP}$, let $a_\lambda\in\{0,1\}$ be the unique bit
satisfying
\[
    \|C_\lambda-\mathcal X^{a_\lambda}\|_\diamond
       \leq\negl(\lambda).
\]
There is a negligible function $\mu$ such that
\[
 \Pr\!\left[
   \Verify(1^\lambda,\pk,\pi)=a_\lambda
   :
   \pk\gets\Gen(1^\lambda,C_\lambda),\
   \pi\gets\Prove(1^\lambda,\pk)
 \right]
 \geq 1-\mu(\lambda).
\]
Moreover, for every possibly nonuniform quantum polynomial-time prover
$\calA$, there is a negligible function
$\mu_{\calA,\{C_\lambda\}}$ such that
\[
 \Pr\!\left[
   \Verify\!\left(
      1^\lambda,\pk,
      \calA(1^\lambda,\pk)
   \right)=1-a_\lambda
   :
   \pk\gets\Gen(1^\lambda,C_\lambda)
 \right]
 \leq\mu_{\calA,\{C_\lambda\}}(\lambda).
\]
The first property is \emph{completeness}; the second is
\emph{wrong-output soundness}.  A malicious prover is allowed to make the
verifier reject or output the correct bit.
\end{definition}

For a publicly verifiable protocol, we define \emph{computational blindness
up to circuit shape and answer} by requiring a \ppt simulator
$\mathsf{Sim}_{\mathrm{pub}}$ such that, for every polynomial-size family
$\{C_\lambda\}_{\lambda\in\NN}$ in $\calC_{\BQP}$ with semantic bits
$a_\lambda$,
\[
    \pk_\lambda\approx_c
       \mathsf{Sim}_{\mathrm{pub}}(1^\lambda,\shape(C_\lambda),a_\lambda),
    \qquad
    \pk_\lambda\gets\Gen(1^\lambda,C_\lambda),
\]
against possibly nonuniform \qpt distinguishers. Thus two circuit families
with the same shape and answer yield indistinguishable prover views.
The answer is unavoidable leakage: an honest prover can produce a proof
from $\pk$ and run the public verifier to recover $a_\lambda$ with
overwhelming probability. Hence this guarantee differs from the private
protocol's blindness, which also hides the answer.

\section{Quantum advantage from Quantum iO and
\texorpdfstring{$\BQP$}{BQP}-hardness}\label{sec:poq}

In this section, we construct a proof of quantumness protocol. \Cref{subsec:private-poq} gives a privately verifiable protocol, and \Cref{subsec:public-poq} lifts this protocol to a publicly verifiable protocol under the same assumptions.
We present both protocols, starting with the simpler privately verifiable construction and then introducing the additional ingredient needed for public verifiability.

\begin{notation}\label{notation:poq}
    We will use the following notation throughout this section.
    \begin{itemize}
        \item Let $\qio$ be a quantum indistinguishability obfuscator for $\calC_{\BQP}$ secure against classical adversaries.
        \item Let $L\in\BQP\setminus\mathrm{i.o.}\BPP$ be a language decided by the uniformly generated circuits $\{V_\lambda\}_{\lambda \in \NN}$. That is, for every input $x \in \{0,1\}^\lambda$, $V_{\lambda}$ outputs $L(x)$ at the designated output bit with error at most $\negl(\lambda)$. Let $V_{\lambda}$ have shape $(s, n_{\inp}, n_{\anc}, n_{\out})$.
        \item For $b\in\{0,1\}$, define $C_{b,\lambda}$ to be a circuit of shape $(2s + 2 \lambda + 1, 1, \lambda + n_{\anc}, 1)$, such that it applies $\sfX^b$ to the first input qubit and acts trivially on the remaining input and ancilla qubits. Let the first qubit be the designated output qubit. Pad $C_{b,\lambda}$ so that it has the required gate count.
        \item We will denote the number of parallel repetitions by $t := t(\lambda)$, and will set $t \le \poly(\lambda)$, $t \ge \omega(\log \lambda)$ to ensure negligible soundness error.
    \end{itemize}
\end{notation}

\subsection{Privately verifiable Proof of Quantumness}\label{subsec:private-poq}
\Cref{fig:poq} describes our construction.
\begin{figure}[H]
\begin{boxedalgo}{\linewidth}
{\centering\bfseries (Privately verifiable) Proof of Quantumness\par}
\medskip
\begin{itemize}
    \item $\Challenge(1^\lambda;r)$:
    \begin{itemize}
        \item Parse $r$ as independent uniform bits $b_1,\ldots,b_t$ followed by random tapes $r_1,\ldots,r_t\in\{0,1\}^{\poly(\lambda)}$.
        \item For all $i\in[t]$, compute
        \[
            \widehat C_i:=
            \qio(1^\lambda,C_{b_i,\lambda};r_i).
        \]
        \item Output $\chall=(\widehat C_1,\ldots,\widehat C_t)$.
    \end{itemize}

    \item $\Prove(1^\lambda,\chall)$:
    \begin{itemize}
        \item Parse $\chall = (\widehat C_1,\ldots,\widehat C_t)$.
        \item For every $i \in [t]$, run the circuit $\widehat{C}_i$ on $\ket{0}_{\inp, \anc}$ state and measure the output register in the computational basis to get $b'_i \in \{0,1\}$.
        \item Output the response $(b_1', \ldots, b_t')$.
    \end{itemize}

    \item $\Verify(1^\lambda,r,\chall,\resp)$:
    \begin{itemize}
        \item Recover $(b_1,\ldots,b_t)$ from $r$ and parse $\resp = (b_1', \ldots, b_t')$.
        \item Accept if $b_i' = b_i$ for every $i \in [t]$, and otherwise reject.
    \end{itemize}
\end{itemize}
\smallskip
\end{boxedalgo}
\par\medskip
\caption{}
\label{fig:poq}
\end{figure}

To show that \Cref{fig:poq} is a sound proof of quantumness, for technical reasons, we need something slightly stronger than $\BQP \nsubseteq \BPP$; we assume that $\BQP\not\subseteq\mathrm{i.o.}\BPP$, where $\mathrm{i.o.}\BPP$ is shorthand for ``infinitely-often'' $\BPP$, and is defined below. The stronger assumption is is needed to handle the mismatch between how we define security of a cryptosystem and correctness of a $\BPP$ machine---an adversary that breaks a proof of quantumness with non-negligible probability needs to win with noticeable probability on only infinitely many input lengths, but a $\BPP$ machine is defined to work for all input lengths.

\begin{definition}[Infinitely-often $\BPP$]\label{def:iobpp}
A language $L$ is in $\mathrm{i.o.}\BPP$ if there is a probabilistic polynomial-time algorithm $M$ such that, for infinitely many input lengths $\lambda$, every $x\in\{0,1\}^\lambda$ satisfies
\[
    \Pr[M(x)=L(x)]\geq \frac23.
\]
Thus $\BQP\not\subseteq\mathrm{i.o.}\BPP$ is stronger than
$\BQP\neq\BPP$.

\end{definition}

\begin{theorem}\label{thm:qio-implies-poq}
Suppose that $\BQP\not\subseteq\mathrm{i.o.}\BPP$ and that there is a qiO scheme for $\calC_{\BQP}$ secure against classical polynomial-time adversaries. Then, the protocol in \Cref{fig:poq} is a secure two-message proof of quantumness.
\end{theorem}

We first prove in \Cref{lem:qio-I-X-ind} that the obfuscations of $C_{0, \lambda}$ and $C_{1, \lambda}$ are computationally indistinguishable under qiO, assuming the worst-case classical hardness of $\BQP$. Then, we prove \Cref{thm:qio-implies-poq} by showing that parallel repetition amplifies hardness so that a classical \ppt adversary can win the spoof the proof of quantumness only with negligible probability.

\begin{lemma}\label{lem:qio-I-X-ind}
Suppose the hypotheses for \Cref{thm:qio-implies-poq} hold. Then, the circuit families $\{C_{0,\lambda}\}_{\lambda\in\NN}$ and
$\{C_{1,\lambda}\}_{\lambda\in\NN}$ defined in \Cref{notation:poq} are in $\calC_{\BQP}$, and efficiently constructible by a classical \ppt algorithm, have the same shape, and induce the channels $\mathsf{I}(\cdot)$ and $\mathsf{X}(\cdot)$ respectively. Further, for every classical \ppt adversary $\calA$ and every $\lambda \in \NN$,
\begin{align}
 \Bigl|
   \Pr\!\left[
     \calA(1^\lambda,\qio(1^\lambda,C_{0,\lambda}))=1
   \right]
   -
   \Pr\!\left[
     \calA(1^\lambda,\qio(1^\lambda,C_{1,\lambda}))=1
   \right]
 \Bigr|
 \leq \negl(\lambda).
 \label{eq:qio-endpoint-hiding}
\end{align}
\end{lemma}

\begin{proof}
Define the language $L$ and circuits $V_\lambda$ as in \Cref{notation:poq}. Standard amplification and deferred measurement yield the required polynomial-time uniform family $\{V_{\lambda}\}_{\lambda \in \NN}$ from a $\BQP$ algorithm for $L$. For each $x\in\{0,1\}^\lambda$, define a circuit $D_x$ with one logical input qubit $Q$
and clean workspace as follows:
\begin{itemize}
    \item prepare $x$ in an instance register, using one gate (possibly an
    identity gate) for each bit of $x$;
    \item apply $V_\lambda$ and, controlled on its output qubit, apply
    $\sfX$ to $Q$;
    \item apply $V_\lambda^\dagger$ and undo the preparation of $x$.
\end{itemize}

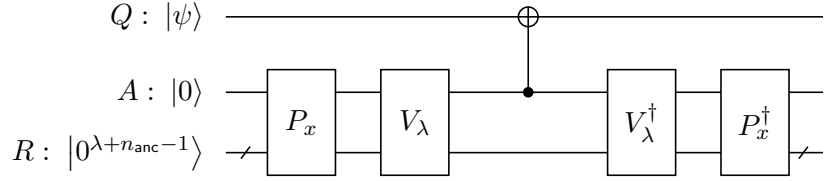
\begin{figure}[t]
\centering
\begin{tikzpicture}[
    x=1cm, y=1cm,
    line width=0.6pt,
    gate/.style={draw, fill=white, minimum width=9mm,
                 minimum height=14mm, inner sep=2pt}
]
    \node[anchor=east] at (-0.15,1.8) {$Q:\ \ket{\psi}$};
    \node[anchor=east] at (-0.15,0.8) {$A:\ \ket{0}$};
    \node[anchor=east] at (-0.15,0)
        {$R:\ \ket{0^{\lambda+n_{\anc}-1}}$};

    \foreach \y in {1.8,0.8,0}
        \draw (0,\y) -- (7.9,\y);

    \draw (0.20,-0.08) -- (0.32,0.08);
    \draw (7.58,-0.08) -- (7.70,0.08);

    \node[gate] at (1.0,0.4) {$P_x$};
    \node[gate] at (2.5,0.4) {$V_\lambda$};
    \node[gate] at (5.5,0.4) {$V_\lambda^\dagger$};
    \node[gate] at (7.0,0.4) {$P_x^\dagger$};

    \draw (4.0,0.8) -- (4.0,1.8);
    \fill (4.0,0.8) circle (2pt);
    \draw[fill=white] (4.0,1.8) circle (0.13);
    \draw (3.87,1.8) -- (4.13,1.8);
    \draw (4.0,1.67) -- (4.0,1.93);
\end{tikzpicture}
\caption{The circuit $D_x$. The qubit $Q$ is its input and designated
output. The $\lambda+n_{\anc}$ workspace qubits start in $\ket{0}_W$,
where $W=AR$, $A$ is the designated output qubit of $V_\lambda$,
and $R$ contains the remaining workspace qubits. Here
$P_x=(\bigotimes_{j=1}^{\lambda}\sfX^{x_j})\otimes I_{\anc}$
prepares $x$ in the instance register and acts trivially on the
original ancillas of $V_\lambda$. The CNOT is controlled by $A$
and targets $Q$. After uncomputation, $D_x$ applies
$\sfX^{L(x)}$ to $Q$ and restores $W$ to $\ket{0}_W$, up to
negligible error. The circuit has $2s+2\lambda+1$ gates, counting
one preparation gate per bit of $x$, including identities.}
\label{fig:Dx}
\end{figure}

Each $D_x$ has $2|V_\lambda|+2\lambda+1$ gates and uses all the input and work qubits of $V_\lambda$ as its clean workspace, so it has the same shape as $C_{0,\lambda}$ and $C_{1,\lambda}$ in \Cref{notation:poq}. By the negligible error guarantee of $V_{\lambda}$, we know that
\begin{align}
    \bigl\|D_x (\cdot)-\sfX^{L(x)}(\cdot)\bigr\|_\diamond
       \leq \negl (\lambda),
\end{align}
and therefore every $D_x$ belongs to $\calC_{\BQP}$ and is negligibly close to $C_{L(x),\lambda}$. Suppose, toward a contradiction, that a classical polynomial-time $\calA$ that distinguishes the $\qio$ obfuscations of $C_{0, \lambda}$ and $C_{1, \lambda}$ with non-negligible advantage. Defining
\[
    p_b(\lambda)
      :=\Pr\!\left[
          \calA(1^\lambda,\qio(1^\lambda,C_{b,\lambda}))=1
        \right]
    \quad\text{and}\quad
    q_x(\lambda)
      :=\Pr\!\left[
          \calA(1^\lambda,\qio(1^\lambda,D_x))=1
        \right],
\]
this means that there is a non-negligible function $\mathsf{non}\text{-}\negl$ such that $|p_0(\lambda)-p_1(\lambda)|\geq \mathsf{non}\text{-}\negl(\lambda).$
Further define
\begin{align}
    \gamma(\lambda)
      :=\max_{x\in\{0,1\}^{\lambda}}
         |q_x(\lambda)-p_{L(x)}(\lambda)|
    \label{eq:qio-max-bridge}
\end{align}
We claim that $\gamma(\lambda) \le \negl(\lambda)$. Indeed, otherwise one could select a maximizing $x_\lambda$ on infinitely many lengths and obtain two same-shape, negligibly close circuit ensembles that $\calA$ distinguishes, contradicting the security of $\qio$.

Now, we give a classical \ppt algorithm $\calB$ that should be able to solve $\BQP$ for infinitely many input lengths $\lambda$. The algorithm works as follows: on input $x\in\{0,1\}^\lambda$, it estimates $p_0(\lambda),p_1(\lambda)$, and $q_x(\lambda)$ using fresh obfuscations and runs of $\calA$, up to sufficiently small inverse polynomial additive accuracy, and outputs $b$ such that the estimate of $q_x(\lambda)$ is closer to $p_b(\lambda)$ than $p_{1-b}(\lambda)$. Polynomially many samples suffice by a Chernoff bound. Therefore, $\calB$ is correct on every input of length $\lambda$, for infinitely many input lengths $\lambda$, which means that $L\in\mathrm{i.o.}\BPP$. This is a contradiction.
\end{proof}

\begin{proof}[Proof of \Cref{thm:qio-implies-poq}]
Correctness holds, since $C_{b, \lambda}$ maps $\ket{0}$ to $\ket{b}$ and the honest prover recovers all the $b_i$ correctly, except with negligible probability.

To prove soundness, we construct a hybrid argument and invoke \Cref{lem:qio-I-X-ind}. Concretely, for $i=0,\ldots,t$, we define hybrid $H_i$ as sampling  $b_1,\ldots, b_t\gets\{0,1\}^t$ independently, and computing the challenge as follows: for every coordinate $j \le i$, compute an independent obfuscation of $C_{0,\lambda}$, i.e., $\widehat{C}_j \gets \qio(1^\lambda, C_{0,\lambda})$, and for coordinates $j > i$, compute the obfuscations as before $\widehat{C}_j \gets \qio(C_{b_j,\lambda})$. We then run $\calA$ on the resulting challenge and accept when its response is exactly $(b_1, \ldots, b_t)$, and reject otherwise.

By \Cref{lem:qio-I-X-ind}, since $t = \poly(\lambda)$, the success probabilities in $H_0$ and $H_t$ differ by at most a negligible function. In hybrid $H_t$, the challenge is independent of the uniform bits $(b_1, \ldots, b_t)$ used for verification, so even an unbounded prover returns $(b_1, \ldots, b_t)$ with probability at most $2^{-t}$. Therefore, for any $t = \omega(\log \lambda)$, soundness holds.
\end{proof}

\subsection{Publicly verifiable proofs of quantumness}
\label{subsec:public-poq}

We next strengthen \Cref{fig:poq} to make it \emph{publicly verifiable}. That is, anyone who sees the challenge and response can check acceptance without the verifier's private coins. We achieve public verifiability via a simple transformation---the challenger publishes the image of the correct response vector under a one-way function. \Cref{fig:poq-public} describes this protocol.

We note that the one-way function needs to be secure only against classical adversaries. Interestingly, as shown in \Cref{lem:qio-implies-classical-owf} below, it turns out that the existence of such classically secure one-way functions is already implied by the assumptions in \Cref{thm:qio-implies-poq}, so we get public verification for free, without having to make any additional assumptions.

\begin{lemma}\label{lem:qio-implies-classical-owf}
Suppose that $\BQP\not\subseteq \mathrm{i.o.}\BPP$ and that there exists an obfuscation scheme for $\calC_{\BQP}$ satisfying qiO security against classical polynomial-time adversaries. Then, there exist one-way functions that are secure against classical \ppt algorithms.
\end{lemma}

\begin{proof}
For $b\in\{0,1\}$, define the families of distributions $\{D_{b,\lambda}\}_{\lambda \in \NN}$ as follows:
\[
    D_{b,\lambda} :=\qio(1^\lambda,C_{b,\lambda}),
\]
where $C_{b, \lambda}$ is defined \Cref{notation:poq}. By \Cref{lem:qio-I-X-ind}, $D_{0,\lambda}$ and $D_{1,\lambda}$ are computationally indistinguishable to classical \ppt adversaries.
On the other hand, by correctness of the obfuscator $\qio$, these distributions are statistically far, with total variation distance at least $1 - \negl(\lambda)$. By Theorem 2.8 of \cite{Goldreich1990CI}, this implies that there exist classically secure pseudorandom generators, and therefore classically secure one-way functions.
\end{proof}

Since the existence of classically secure one-way functions implies that $\mathsf{P} \neq \mathsf{NP}$, we immediately get the following corollary that highlights the role of (quantum) iO as a hardness ``transformer''.

\begin{corollary}
    Suppose that $\BQP\not\subseteq \mathrm{i.o.}\BPP$ and that there exists an obfuscation scheme for $\calC_{\BQP}$ satisfying qiO security against classical polynomial-time adversaries. Then, $\mathsf{P} \neq \mathsf{NP}$.
\end{corollary}

Let $\{f_{\lambda}\}_{\lambda \in \NN}$ be a one-way function family secure against classical \ppt algorithms, where $f_\lambda$ has input length $\lambda$. Let $t = \poly(\lambda)$. We will write $f : = f_t$ as shorthand. We describe our construction \Cref{fig:poq-public}.

\begin{figure}[H]
\begin{boxedalgo}{\linewidth}
\label{fig:public-poq}
{\centering\bfseries Publicly verifiable Proof of Quantumness\par}
\medskip
\begin{itemize}
    \item $\Verify_1(1^\lambda;r)$:
    \begin{itemize}
        \item Parse $r$ as $\vecb := (b_1,\ldots,b_t)\in\{0,1\}^t$ followed by random tapes $r_1,\ldots,r_t\in\{0,1\}^{\poly(\lambda)}$.
        \item For all $i \in [t]$, compute
        \[
            \widehat C_i:=
            \qio(1^\lambda,C_{b_i,\lambda};r_i)
            \qquad(i\in[t]).
        \]
        \item Compute $y=f(\vecb)$ and output
        \[
            \chall=(\widehat C_1,\ldots,\widehat C_t,y).
        \]
    \end{itemize}

    \item $\Prove(1^\lambda,\chall)$:
    \begin{itemize}
        \item Parse the challenge $\chall = (\widehat{C}_1, \ldots, \widehat{C}_t, y)$.
        \item Evaluate each $\widehat C_i$ on $\ket{0}$, measure its logical output
        qubit in the computational basis, and output the resulting string
        $\resp=(b'_1,\ldots,b'_t)$.
    \end{itemize}

    \item $\Verify(1^\lambda,\chall,\resp)$:
    \begin{itemize}
        \item Parse the challenge as
        $\chall = (\widehat C_1,\ldots,\widehat C_t,y)$ and the response as
        $\resp\in\{0,1\}^t$.
        \item Accept if $f(\resp)=y$ and reject otherwise.
    \end{itemize}
\end{itemize}
\smallskip
\end{boxedalgo}
\par\medskip
\caption{}
\label{fig:poq-public}
\end{figure}

\begin{theorem}\label{thm:public-poq}
Suppose that $\BQP\not\subseteq\mathrm{i.o.}\BPP$, and that there exists an obfuscation scheme for $\calC_{\BQP}$ satisfying qiO security against classical polynomial-time adversaries. Then \Cref{fig:public-poq} is a secure publicly verifiable two-message proof of quantumness.
\end{theorem}

\begin{proof}
Completeness holds by the correctness of the qiO scheme.

We prove soundness using a slightly more careful hybrid argument compared to the proof of \Cref{thm:qio-implies-poq}.
For $i= 0, 1, \ldots, t-1$, we define the following hybrids, which differ only in the circuit at index $i+1$:
\begin{itemize}
    \item Hybrid $H_{2i}$:
        \begin{itemize}
        \item Sample a uniformly random $\vecb \gets \{0,1\}^t$ and compute $y := f(\vecb)$.
        \item Compute the obfuscated circuits $\hat{C}_1, \ldots, \hat{C}_t$ as follows:
        \begin{align*}
            \widehat{C}_j &\gets \qio(1^\lambda, C_{0, \lambda}) &\text{for all }j \le i\\
            \hat{C}_{i+1} &\gets \qio(1^\lambda, C_{b_{i+1}, \lambda})\\
            \hat{C}_j &\gets \qio(1^\lambda, C_{b_j, \lambda}) &\text{for all }j > i+1
        \end{align*}
        \item Given a response $\vecb'$, the prover wins if $f(\vecb') = y$.
        \end{itemize}
    \item Hybrid $H_{2i+1}$:
        \begin{itemize}
        \item Sample a uniformly random $\vecb \gets \{0,1\}^t$ and compute $y := f(\vecb)$.
        \item Compute the obfuscated circuits $\hat{C}_1, \ldots, \hat{C}_t$ as follows:
        \begin{align*}
            \widehat{C}_j &\gets \qio(1^\lambda, C_{0, \lambda}) &\text{for all }j \le i\\
            \hat{C}_{i+1} &\gets \qio(1^\lambda, C_{1, \lambda})\\
            \hat{C}_j &\gets \qio(1^\lambda, C_{b_j, \lambda}) &\text{for all }j > i+1
        \end{align*}
        \item Given a response $\vecb'$, the prover wins if $f(\vecb') = y$.
        \end{itemize}
\end{itemize}
Hybrid $2t$ is defined as having sampled $\vecb$ and $y$ as every other hybrid, and all circuits being obfuscations of $C_{0,\lambda}$.
For a classical \ppt adversary $\calA$, let $p_n$ be its success probability for hybrid $n$.

We claim that $|p_{2i} - p_{2i+1}|\le \negl(\lambda)$ for all $i = 0, 1, \ldots, t-1$. Suppose for contradiction there is a \ppt distinguisher. Then, we have a distinguisher that contradicts \Cref{lem:qio-I-X-ind}, via the following reduction: given input $\hat{C}$, sample random $\vecb \in \{0,1\}^t$, compute $y := f(\vecb)$, and compute the challenge circuits $\hat{C}_j$ for $j \neq i+1$ to be consistent with both hybrids $H_{2i}$ and $H_{2i+1}$. At the $i+1$th index, set $\hat{C}_{i+1}$ as follows:
\begin{align*}
    \hat{C}_{i+1} &\gets
    \begin{cases}
    \hat{C} &\text{if } b_{i+1} = 0,\\
    \qio(1^\lambda, C_{1, \lambda}) &\text{if } b_{i+1} = 1.
    \end{cases}
\end{align*}
When the challenge $\hat{C}$ obfuscates $C_{0, \lambda}$, this exactly simulates $H_{2j}$, and when the challenge obfuscates $C_{1, \lambda}$, this exactly simulates $H_{2j+1}$. Therefore, by \Cref{lem:qio-I-X-ind}, we have that $|p_{2i} - p_{2i+1}|\le \negl(\lambda)$.

A direct reduction and application of \Cref{lem:qio-I-X-ind} also shows that $|p_{2j+1} - p_{2j+2}|\le \negl(\lambda)$. By the triangle inequality, and since we have only polynomially many hybrids, this means that $|p_0 - p_{2t}| \le \negl(\lambda)$. In Hybrid $H_{2t}$, the obfuscated circuits in the challenge are independent of the one-way function challenge, so by the one-way security of $f$ against classical adversaries, $p_{2t} \le \negl(\lambda)$. This completes the proof.
\end{proof}

\section{Classical Verification of $\BQP$ computation}
\begin{notation}\label{notation:cvqc}
    We use the following notation throughout this section.
    \begin{itemize}
        \item Let $\rows, \cols \in \NN$ be efficiently computable from $\lambda$, with $\lambda \le \rows \le \poly(\lambda)$ and $2 \le \cols \le \poly(\lambda)$.
        \item Let $C$ be the $\BQP$-circuit being verified. Without loss of generality, we assume that $C$ has been amplified, so that it produces the desired output on $\ket{0}$ with negligible trace distance. We also assume without loss of generality that it returns a $\ket{0}$ ancilla workspace to itself, i.e., it is a tidy computation. Therefore $C \in \calC_{\BQP}$.
        \item For $b \in \{0,1\}$, let $D_b$ be a circuit implementing the channel $\sfX^b(\cdot)$, with the same number of input, ancilla, and output wires as $C$. Further, let $D_b$ be padded to have the same gate count as $\sfX \circ C$ and $\sfI \circ C$. Choose the descriptions and padding of $D_b$ canonically from $b$ and this common shape alone.
        \item Define $\{g_{\lambda}\}_{\lambda \in \NN}$ to be a family of quantum-secure one-way functions.
    \end{itemize}
\end{notation}

\subsection{Private verification assuming qiO}
\label{sec:private-bqp-verification}

\begin{figure}
\begin{boxedalgo}{\linewidth}
\label{fig:private-two-codeword}
{\centering\bfseries Privately verifiable Classical Verification of $\BQP$\par}
\medskip
\begin{itemize}
    \item $\Gen(1^\lambda,C)$:
    \begin{itemize}
        \item Sample independent uniform mask bits $\vecb=(b_{i,j})_{i\in[\rows],\,j\in[\cols]}$.
        \item Independently sample planted locations $j_i^*\gets[\cols]$ for every row $i\in[\rows]$, and let $J^*=\{(i,j_i^*):i\in[\rows]\}$.
        \item For every $i\in[\rows]$ and $j\in[\cols]$, independently compute
        \[
            \widehat C_{i,j}\gets
            \begin{cases}
                \qio(1^\lambda,D_{b_{i,j}}), & j\neq j_i^*,\\
                \qio(1^\lambda,\sfX^{b_{i,j}}\circ C), & j=j_i^*.
            \end{cases}
        \]
        \item Output the challenge $\chall=(\widehat C_{i,j})_{i,j}$ and retain the private state $\mathsf{state}=(J^*,\vecb)$.
    \end{itemize}

    \item $\Prove(1^\lambda,\chall)$:
    \begin{itemize}
        \item Parse $\chall=(\widehat C_{i,j})_{i,j}$.
        \item Evaluate each $\widehat C_{i,j}$ on $\ket{0}$ and measure its output in the computational basis, obtaining $b'_{i,j}$.
        \item Return $\vecb'=(b'_{i,j})_{i,j}$.
    \end{itemize}

    \item $\Verify(1^\lambda,C,\chall,\mathsf{state},\vecb')$:
    \begin{itemize}
        \item Parse $\mathsf{state}=(J^*,\vecb)$ and $\vecb'=(b'_{i,j})_{i,j}\in\{0,1\}^{\rows\cdot\cols}$.
        \item Check that $b'_{i,j} = b_{i,j}$ for all test location $(i,j) \notin J^*$, and reject otherwise.
        \item If the bits $b'_{i,j}\oplus b_{i,j}$ for all the planted locations $(i,j) \in J^*$ are all equal, output their common value; otherwise reject.
    \end{itemize}
\end{itemize}
\smallskip
\end{boxedalgo}
\par\medskip
\caption{}
\label{fig:private-cvqc}
\end{figure}

\begin{theorem}
\label{thm:private-bqp-verification}
Suppose there is an obfuscation scheme for $\calC_{\BQP}$ satisfying qiO security against \qpt adversaries. Then, \Cref{fig:private-two-codeword} is a privately verifiable protocol for classical verification of $\BQP$ that is computationally blind up to circuit shape.
\end{theorem}

\begin{proof}
Let $a$ be the output of $C$ on $\ket{0}$. Let $\mathbf{1}_{J^*}$ be the indicator vector for the set of planted locations $J^*$, indexed by $i \in [\rows], j \in [\cols]$. We will write $\vecs := \vecb\oplus a \mathbf{1}_{J^*}$.

We first show that the protocol is correct. at each test location $(i,j)\notin J^*$, $D_{b_{i,j}}$ implements $\sfX^{s_{i,j}}$ exactly, and at each planted location $(i,j) \in J^*$,
\begin{align*}
    \left\|\sfX^{b_{i,j}}\circ C-\sfX^{s_{i,j}}\right\|_\diamond
       \leq\negl(\lambda).
\end{align*}
Thus every circuit submitted to $\qio$ lies in $\calC_{\BQP}$, and by the correctness of $\qio$, an honest prover measures $\vecb' = \vecs$ except with negligible probability. All test answers $b_{i,j}'$ therefore match the $b_{i,j}$ values, and all unmasked planted answers equal $a$, proving completeness.

We now prove soundness via a hybrid argument. For $i=0,\ldots,\rows$, define $H_i$ by replacing the planted circuit in each of the first $i$ rows by an independently generated obfuscation of $D_{s_{k,j_k^*}}$, for $k\leq i$. Test locations and the secret verification state remain as in the real experiment. Let $p_i$ be the probability that the private verifier outputs $1-a$ in $H_i$. By qiO security, it must hold that $|p_i - p_{i+1}| \le \negl(\lambda)$ for all $i=0,\ldots,\rows-1$, since $D_{s_{k,j_k^*}}$ and $\sfX^{b_{i,j}}\circ C$ are negligibly close in diamond distance. Since $\rows \le \poly(\lambda)$, a triangle inequality means that $|p_0-p_{\rows}|$ is negligible.

In $H_{\rows}$, every challenge entry is a fresh obfuscation of $D_{s_{i,j}}$, where $\vecs$ is uniform and independent of $J^*$. This means that the challenge, and therefore the prover's response, is is independent of $J^*$.
If the verifier outputs $1-a$, this means that the prover produced the output $\vecb' := \vecb \oplus (1-a)\mathbf{1}_{J^*} = \vecs \oplus \mathbf{1}_{J^*}$, which can happen with probability at most $1 / \cols^{\rows}$, therefore $p_{\rows} \le 1 / \cols^{\rows} \le 2^{-\lambda}$. This proves soundness.

The same hybrids establish blindness. Given only $1^\lambda$ and
$\shape(C)$, a classical simulator constructs the canonical padded
circuits $D_0,D_1$, samples independent uniform bits $s_{i,j}$, and outputs
\[
    \bigl(\qio(1^\lambda,D_{s_{i,j}})\bigr)_{i,j},
\]
using independent obfuscation randomness at every entry. This is exactly
the challenge distribution in $H_{\rows}$: the uniform masks make
$\vecs=\vecb\oplus a\mathbf{1}_{J^*}$ uniform, regardless of $C$ or its
answer $a$. The qiO hybrid argument makes the real challenge
computationally indistinguishable from this simulated distribution.
Since the prover receives only the challenge and no later verifier
message, any quantum polynomial-time processing of that challenge also
has an indistinguishable simulated view. Thus the protocol hides the
delegated circuit, including its hardcoded instance and answer, up to
circuit shape. Padding to common public gate-count and width bounds can
make the leakage depend only on those bounds.
\end{proof}

\subsection{Public verification assuming qiO and quantum-secure OWFs}
\label{sec:public-bqp-verification}

For this protocol, set $\rows=\lambda$ and $\cols=2$, and write
$g=g_\lambda:\{0,1\}^{\lambda}\to\{0,1\}^{m(\lambda)}$ for the
quantum-secure one-way function, where $m$ is polynomially bounded.

The challenge contains the same grid as in \Cref{fig:private-cvqc}, together with the parity of the mask bits in each row and one-way images of the first columns of the two grids accepted by the private verifier. The row parities determine the claimed answer, and the one-way images let a public verifier check that answer.

\begin{figure}
\begin{boxedalgo}{\linewidth}
{\centering\bfseries Publicly verifiable Classical Verification of $\BQP$\par}
\medskip
\begin{itemize}
    \item $\Gen(1^\lambda,C)$:
    \begin{itemize}
        \item Sample independent uniform mask bits $\vecb=(b_{i,j})_{i\in[\rows],\,j\in[2]}$.
        \item Independently sample planted locations $j_i^*\gets[2]$ for every row $i\in[\rows]$, and let $J^*=\{(i,j_i^*):i\in[\rows]\}$.
        \item For every $i\in[\rows]$ and $j\in[2]$, independently compute
        \[
            \widehat C_{i,j}\gets
            \begin{cases}
                \qio(1^\lambda,D_{b_{i,j}}), & j\neq j_i^*,\\
                \qio(1^\lambda,\sfX^{b_{i,j}}\circ C), & j=j_i^*.
            \end{cases}
        \]
        \item Let $\vecs^0:=\vecb$ and $\vecs^1:=\vecb\oplus\mathbf{1}_{J^*}$, where $\mathbf{1}_{J^*}$ is the indicator vector for the planted locations. Set $d_i:=b_{i,1}\oplus b_{i,2}$ for every row, and compute
        \[
            y_0:=g\bigl((s^0_{i,1})_{i\in[\rows]}\bigr),
            \qquad
            y_1:=g\bigl((s^1_{i,1})_{i\in[\rows]}\bigr).
        \]
        \item Output the challenge $\chall=((\widehat C_{i,j})_{i,j},(d_i)_{i\in[\rows]},y_0,y_1)$.
    \end{itemize}

    \item $\Prove(1^\lambda,\chall)$:
    \begin{itemize}
        \item Parse $\chall=((\widehat C_{i,j})_{i,j},(d_i)_{i\in[\rows]},y_0,y_1)$.
        \item Evaluate each $\widehat C_{i,j}$ on $\ket{0}$ and measure its output in the computational basis, obtaining $b'_{i,j}$.
        \item Return $\vecb'=(b'_{i,j})_{i,j}$.
    \end{itemize}

    \item $\Verify(1^\lambda,\chall,\vecb')$:
    \begin{itemize}
        \item Parse $\chall$ as above and $\vecb'=(b'_{i,j})_{i,j}\in\{0,1\}^{\rows\cdot2}$.
        \item If the bits $b'_{i,1}\oplus b'_{i,2}\oplus d_i$ for all rows $i\in[\rows]$ are not all equal, reject. Otherwise let $b$ be their common value.
        \item If $g\bigl((b'_{i,1})_{i\in[\rows]}\bigr)=y_b$, output $b$; otherwise reject.
    \end{itemize}
\end{itemize}
\smallskip
\end{boxedalgo}
\par\medskip
\caption{}
\label{fig:public-cvqc}
\label{con:two-codeword-cvqc}
\end{figure}

\begin{theorem}\label{thm:public-bqp-verification}
Suppose there is an obfuscation scheme for $\calC_{\BQP}$ satisfying qiO security against nonuniform \qpt adversaries, and that $g$ is a one-way function secure against such adversaries. Then \Cref{fig:public-cvqc} is a publicly verifiable protocol for classical verification of circuits in $\calC_{\BQP}$, with computational blindness up to circuit shape and answer.
\end{theorem}

\begin{proof}
Fix $C=C_\lambda\in\calC_{\BQP}$ and let $a$ be the output of $C$ on $\ket{0}$. As in \Cref{sec:private-bqp-verification}, let $\mathbf{1}_{J^*}$ be the indicator vector for the planted locations, and write $\vecs:=\vecs^a = \vecb\oplus a\mathbf{1}_{J^*}$.

We first show that the protocol is correct. At each test location $(i,j)\notin J^*$, $D_{b_{i,j}}$ implements $\sfX^{s_{i,j}}$ exactly, and at each planted location $(i,j)\in J^*$, we have that $\left\|\sfX^{b_{i,j}}\circ C-\sfX^{s_{i,j}}\right\|_\diamond \leq\negl(\lambda)$. By correctness of $\qio$, an honest prover returns $\vecb'=\vecs$ except with negligible probability. Since each row has exactly one planted location, $s_{i,1}\oplus s_{i,2}\oplus d_i=a$ for every row. The first column also has image $y_a$, so the verifier does not reject, and outputs $a$.

We now prove soundness via a hybrid argument. As in the proof of \Cref{thm:private-bqp-verification}, we define $H_i$ for $i=0,\ldots,\rows$, by replacing the planted circuit in each of the first $i$ rows by an independently generated obfuscation of $D_{s_{k,j_k^*}}$, for $k\leq i$. Let $p_i$ be the probability that the verifier outputs $1-a$ in $H_i$. The qiO security guarantee against nonuniform \qpt adversaries implies that $|p_0 - p_{\rows}| \le \negl(\lambda)$: the correlated classical public data can be fixed and hardwired into a distinguisher.

To complete the proof, we show that $p_{\rows} \le \negl(\lambda)$. In $H_{\rows}$, every challenge circuit is a fresh obfuscation of $D_{s_{i,j}}$, where $\vecs$ is uniform and independent of $J^*$. Since $\cols=2$, each bit $\mathbf{1}_{\{j_i^*=1\}}$ is an independent uniform bit. Consequently, the bits
\[
    s^{1-a}_{i,1}=s_{i,1}\oplus\mathbf{1}_{\{j_i^*=1\}}
\]
form a uniform $\lambda$-bit string independent of the entire grid $\vecs$. The published parities satisfy $d_i=s_{i,1}\oplus s_{i,2}\oplus a$, so they depend only on $\vecs$ and $a$. Therefore, the challenge depends on the $(s_{i,1}^{1-a})_{i \in [\rows]}$ bits only through $y_{1-a}$. Therefore, we can use an adversary that breaks the verification protocol to construct an adversary for the one-way function $g$, by giving a reduction that simulates $H_{\rows}$. The one-way function security of $g$ therefore establishes that $p_{\rows} \le \negl(\lambda)$.

Finally, the same hybrids prove computational blindness up to circuit
shape and answer. Given only $1^\lambda$, $\shape(C)$, and $a$, a classical
simulator samples a uniform grid $\vecs\in\{0,1\}^{\lambda\times2}$ and an
independent uniform string $u\in\{0,1\}^{\lambda}$. It generates independent
obfuscations $\widehat D_{i,j}\gets\qio(1^\lambda,D_{s_{i,j}})$ and sets
\[
    d_i=s_{i,1}\oplus s_{i,2}\oplus a,
    \qquad y_a=g\bigl((s_{i,1})_{i\in[\lambda]}\bigr),
    \qquad y_{1-a}=g(u).
\]
It outputs $((\widehat D_{i,j})_{i,j},(d_i)_i,y_0,y_1)$. The independence
of the opposite first-column string proved above shows that this is
exactly the joint challenge distribution in $H_{\rows}$. Applying the
same qiO hybrids to any quantum polynomial-time distinguisher makes the
real challenge computationally indistinguishable from this simulation.
This proves computational blindness up to circuit shape and semantic
answer. Unlike the private protocol, this protocol cannot
hide the answer: evaluating the two circuits in any row and computing
$b'_{i,1}\oplus b'_{i,2}\oplus d_i$ recovers $a$ with overwhelming
probability. No hardness property of $g$ is needed for this blindness
argument; one-wayness is used for soundness.
\end{proof}

\section{Constructions of Quantum Indistinguishability Obfuscation}
\label{sec:obfuscationprimitives}

In this section we discuss the theoretical foundations for constructing quantum indistinguishability obfuscation. We first consider qiO for ancilla-free unitary circuits and prove a worst-to-average case reduction, which is inspired by an analogous argument for classical reversible circuits in~\cite{Canetti2024Towards}. Then we show that qiO for ancilla-free unitary circuits directly implies qiO for tidy unitary circuits, which provides the foundation of our main applications. This argument is inspired by the reversible classical circuit embedding of Appendix A of the full version of \cite{Canetti2024Towards}, which makes the compiled circuit independent of the original circuit on inputs with incorrectly initialized ancillas. To generalize this construction to quantum circuits, we introduce an additional random phase and a coupling argument (see~\Cref{fig:tidy-ancilla-free-compilation}).

In \Cref{sec:qioassumptions}, we introduce the main relevant definitions and assumptions. We then show the worst-to-average case reduction for qiO for ancilla-free unitary circuits in \Cref{sec:qio-worst-to-average-reduction}. Finally in \Cref{sec:tidy-qio-lifting} we prove that qiO for ancilla-free unitary circuits implies qiO for tidy unitary circuits.

\subsection{Definitions and Assumptions}
\label{sec:qioassumptions}

Recall that an ancilla-free unitary circuit $C$ has $n$ input and output wires and implements an $n$-qubit unitary $U_C$. A tidy unitary circuit (see \Cref{sec:quantum-circuits}) has an $n$-qubit logical register and an $a$-qubit ancilla register, and returns the ancilla register to $\ket{0^a}$, up to negligible error, on every logical input when the ancillas are initialized to $\ket{0^a}$. We consider qiO for these two circuit families, using the definition of qiO from \Cref{def:qio} throughout. The restriction to ancilla-free or tidy circuits describes only the input class of the obfuscator; its output may use internal ancillas and need not itself be tidy.

Next, we discuss assumptions related to quantum random input obfuscation (RIO), which is a generalization of random input obfuscation for classical reversible circuits considered in~\cite{Canetti2024Towards}. Our definition of quantum RIO is conceptually simpler than the definition of classical RIO in~\cite{Canetti2024Towards}: here quantum RIO is a circuit obfuscator that, given the description of an ancilla-free random quantum circuit as input, outputs a circuit that is computationally indistinguishable from a uniformly random circuit drawn from the ensemble of circuits that implements the same unitary as the input circuit. This can be viewed as a weaker, average-case analogue of qiO for ancilla-free unitary circuits. We show in \Cref{sec:qio-worst-to-average-reduction} that quantum RIO, with additional assumptions discussed below, first gives security for exactly equivalent ancilla-free circuits and, under an additional stability assumption, implies qiO as defined in \Cref{def:qio}. This can be viewed as a worst-to-average case reduction for qiO and thus gives evidence supporting the existence of qiO.

The assumptions below posit the existence of a common choice of a fixed finite, efficiently specified, inverse-closed universal gate set containing $I$, $\sfX$, and $\CNOT$, together with the stated computational hardness properties. Whenever these assumptions are combined, they must hold for the same choice of gate set. All circuits and distributions below use this gate set. We use the convention $U_{A|B}=U_B U_A$, where $A|B$ denotes the concatenation of circuit descriptions $A$ and $B$. Write $\calG_n$ for the legal gate placements on $n$ qubits and $\mathsf{Circ}_{n,m}$ for the $m$-gate ancilla-free circuits obtained by choosing each gate from $\calG_n$. Uniform sampling from $\mathsf{Circ}_{n,m}$ chooses these gates independently and uniformly. $|C|$ denotes the gate count of a circuit $C$. For $C\in\mathsf{Circ}_{n,m}$ and $L\geq m$, let $\calE_{C,L}$ be uniform over the $L$-gate circuits implementing the same channel as $C$. Equality of these channels permits a global phase difference between their implementing unitaries. The identity gate ensures that the set sampled by $\calE_{C,L}$ is nonempty. Note that the uniform distribution over $\calE_{C,L}$ may not be efficiently samplable in general.

The first assumption says that a random separator and its inverse can be removed from two independently randomized pieces. It applies even when the surrounding circuits and side information are correlated; the separator alone must be fresh.

\begin{assumption}[Quantum split-circuit pseudorandomness, SCP]
\label{assumption:qSCP}
There are a polynomial-time computable, polynomially bounded integer $s=s(\lambda,n)$ and a polynomially bounded $\mu=\mu(\lambda,n)\geq1$ with the following property. Let $(Q,A,B)$ be any joint polynomial-size ensemble, where $A,B$ are ancilla-free $n$-qubit circuits and $Q$ is arbitrary classical or quantum side information. Independently sample $R\gets\mathsf{Circ}_{n,s}$. For polynomially bounded lengths satisfying
\[
L_A\geq\mu(\lambda,n)(|A|+s+1),\qquad L_B\geq\mu(\lambda,n)(|B|+s+1),
\]
conditionally and independently sample $X\gets\calE_{A|R,L_A}$ and $Y\gets\calE_{R^\dagger|B,L_B}$. Let $Z\gets\calE_{A|B,L_A+L_B}$ be freshly sampled, independently of $R$ given $(Q,A,B)$. Then $(1^\lambda,Q,X|Y)$ and $(1^\lambda,Q,Z)$ are computationally indistinguishable against nonuniform quantum polynomial-time distinguishers. We call lengths satisfying the displayed inequalities SCP-admissible.
\end{assumption}

The second assumption concerns obfuscation of random circuit descriptions. Given a plaintext random circuit input, the obfuscated circuit is indistinguishable from a postprocessed random equivalent circuit. Our definition also allows a correlated public view, which is necessary when a random circuit occurs inside a larger construction.

\begin{definition}[Quantum random input obfuscation, RIO]
\label{def:qRIO}
Let $\Obf,\pi$ be classical probabilistic polynomial-time transformations of ancilla-free circuits over the chosen gate set that preserve the width and implemented channel exactly on every output. Their outputs are ancilla-free, with gate counts depending only on $\lambda,n$ and the input gate count. Let $\xi(\lambda,n,m)\geq m$ be polynomial-time computable and polynomially bounded. The output lengths of $\Obf$ on $m$-gate inputs and $\pi$ on $\xi(\lambda,n,m)$-gate inputs must agree. For a polynomial-size joint ensemble $(V,C)$, where $V$ is a public classical view and $C$ is an ancilla-free circuit, the random-input obfuscation guarantee is
\[
(1^\lambda,V,C,\Obf(1^\lambda,C))\ \approx_c\ (1^\lambda,V,C,\pi(1^\lambda,D)),\qquad D\gets\calE_{C,\xi(\lambda,n,|C|)}.
\]
Here indistinguishability is against nonuniform quantum polynomial-time distinguishers, and $(V,C)$ need not be efficiently samplable. The ideal sample and all algorithmic randomness are fresh and mutually independent conditional on $(V,C)$; the ideal distribution depends on the view only through $C$. The guarantee is \emph{uniform-public-view} at length $m$ if it holds for every such experiment whose marginal distribution of $C$ is uniform over $\mathsf{Circ}_{n,m}$, even when $V$ is correlated with $C$.
\end{definition}

Next we state the assumption which gives the specific form of quantum RIO we use.

\begin{assumption}[Quantum RIO assumption]
\label{assumption:qRIO}
Let $s=s(\lambda,n)$ and $\mu=\mu(\lambda,n)\geq1$ be the parameters from the quantum SCP assumption. There exist $\Obf,\xi,\pi$ as in~\Cref{def:qRIO} with the following properties.

First, they satisfy uniform-public-view quantum RIO at lengths $s$ and $2s+1$. Let $q=q(\lambda,n)$ be the output gate count of $\Obf$ on $n$-qubit, $s$-gate inputs, and define
\[
b:=\xi(\lambda,n,s),\quad g:=\xi(\lambda,n,2s+1),\quad e:=\xi(\lambda,n,2q).
\]
These parameters depend on $(\lambda,n)$, which we suppress below. We require
\[
b\geq\mu(s+1),\qquad g\geq\mu(2s+2),\qquad e\geq\mu(2s+1).
\]

Second, quantum RIO must also hold for the experiment below, which describes joining two obfuscated circuit pieces. To define this experiment, for $m\geq1$ set $M_m:=mg+(m-1)e$. This length accounts for $m$ masked-gate blocks of length $g$ and the $m-1$ joins between them, each of length $e$. Let $\Pi_m$ parse an $M_m$-gate circuit into alternating blocks of lengths $g,e,g,\ldots,e,g$ and apply $\pi$ to each block with independent randomness.

Define the \emph{ideal triple} for an $m$-gate, $n$-qubit circuit $C$ as follows. Sample independent masks $T_L,T_R\gets\mathsf{Circ}_{n,s}$ and, conditional on these masks, independently sample
\[
D_L\gets\calE_{T_L,b},\qquad
D_0\gets\calE_{T_L^\dagger|C|T_R,M_m},\qquad
D_R\gets\calE_{T_R^\dagger,b}.
\]
The ideal triple is
\[
\bigl(\pi(1^\lambda,D_L),\Pi_m(1^\lambda,D_0),\pi(1^\lambda,D_R)\bigr),
\]
with independent postprocessing randomness. The two outer circuits each have $q$ gates by the output-length condition in the RIO definition. Concatenating the three circuits implements the same channel as $C$, since the masks cancel. These ideal samples are used to specify the security experiment and need not be efficiently samplable.

For every pair of polynomial-size $n$-qubit circuit families $A,B$ with positive gate counts, independently sample their ideal triples $(X_L,X_0,X_R)$ and $(Y_L,Y_0,Y_R)$, and set
\[
C_{\mathrm{seam}}=X_R|Y_L,\qquad V=(X_L,X_0,Y_0,Y_R).
\]
Thus $C_{\mathrm{seam}}$ is the $2q$-gate circuit formed by concatenating the two adjacent boundary circuits, and $V$ contains the four surrounding circuits. The same $\Obf,\xi,\pi$ must satisfy the quantum RIO guarantee for this joint ensemble $(V,C_{\mathrm{seam}})$, with ideal sample $D\gets\calE_{C_{\mathrm{seam}},e}$.
\end{assumption}

\subsection{Obfuscation for Ancilla-Free Circuits}
\label{sec:qio-worst-to-average-reduction}

We first construct an obfuscator with indistinguishability guaranteed for circuits implementing exactly the same channel. We then extend this guarantee to qiO for circuits with negligibly close channels, under an additional stability assumption (\Cref{assumption:stability}).

\begin{theorem}[Security for exactly equivalent ancilla-free circuits]
\label{thm:local-mixing-ancilla-free-qio}
Under quantum split-circuit pseudorandomness (\Cref{assumption:qSCP}) and quantum RIO (\Cref{assumption:qRIO}) above, there is a polynomial-time, exactly channel-preserving obfuscator for polynomial-size ancilla-free unitary circuits whose obfuscations are computationally indistinguishable for any two exactly equivalent, same-shape circuits against nonuniform quantum polynomial-time adversaries. Its output length on an $m$-gate, $n$-qubit circuit is $m\cdot\poly(\lambda,n)$.
\end{theorem}

\begin{proof}
Below we suppress the security parameter in algorithm calls, and write $\mathsf I(C)$ for the ideal triple in \Cref{assumption:qRIO}. Return the empty circuit unchanged when $m=0$; below assume $m\geq1$.

\paragraph{Construction.}
We represent an obfuscated circuit by a triple $X=(X_L,X_0,X_R)$ of left, middle, and right blocks. The blocks implement the channels of
\[
L,\qquad L^\dagger|C|R,\qquad R^\dagger,
\]
for random boundary masks $L,R$. Thus $\operatorname{concat}(X):=X_L|X_0|X_R$ implements $C$. Keeping the boundary blocks separate lets us join two such triples by reobfuscating their adjacent boundaries.

Define a recursive sampler $\mathsf S$ as follows. For a single gate $\beta$, sample independent $L,R\gets\mathsf{Circ}_{n,s}$ and set
\[
\mathsf S(\beta):=
\bigl(\Obf(L),\Obf(L^\dagger|\beta|R),\Obf(R^\dagger)\bigr).
\]
For two triples $X,Y$, define
\[
\mathsf{Join}(X,Y):=
\bigl(X_L,\;X_0|\Obf(X_R|Y_L)|Y_0,\;Y_R\bigr).
\]
For a longer circuit $C$, split $C=A|B$ with $|A|=\lfloor |C|/2\rfloor$, independently sample $X\gets\mathsf S(A)$ and $Y\gets\mathsf S(B)$, and return $\mathsf S(C):=\mathsf{Join}(X,Y)$. All obfuscation calls use fresh randomness. The final obfuscator returns $\operatorname{concat}(\mathsf S(C))$.

The correctness of the construction follows from cancellation of the boundary masks. Each join preserves the channel of $\operatorname{concat}(X)|\operatorname{concat}(Y)$, since $\Obf$ preserves the channel of $X_R|Y_L$. This proves exact correctness by induction. The outer blocks always have $q$ gates, so the recursion makes $O(m)$ calls on inputs of lengths $s$, $2s+1$, or $2q$. It runs in polynomial time and outputs $m\cdot\poly(\lambda,n)$ gates.

\paragraph{Security.}
For fixed input shape, the distribution $\calE_{T_L^\dagger|C|T_R,M_m}$, and hence $\mathsf I(C)$, depends on $C$ only through its channel. It therefore suffices to prove the invariant
\[
\mathsf S(C)\approx_c\mathsf I(C).
\]
We establish the single-gate case and then show that joining preserves this invariant.

For a single gate $\beta$, the outer inputs $L,R^\dagger$ are uniformly random, so uniform-public-view RIO applies directly to them. The middle input $L^\dagger|\beta|R$ is a uniform $(2s+1)$-gate circuit conditioned on its center gate being $\beta$, an event of probability $p=1/|\calG_n|\geq1/\poly(n)$. This event is efficiently testable from the plaintext circuit. If a distinguisher had advantage $\Delta$ in the conditioned experiment, a distinguisher for the unconditioned experiment could test this event, run the former distinguisher when it holds, and output $0$ otherwise. Its advantage would be $p\Delta$. Since $p^{-1}$ is polynomial in $\lambda$ for polynomially bounded $n$, a nonnegligible $\Delta$ would contradict RIO. Thus RIO remains valid under this conditioning.

Conditional on $L,R$, independently sample
\[
D_L\gets\calE_{L,b},\qquad
D_0\gets\calE_{L^\dagger|\beta|R,g},\qquad
D_R\gets\calE_{R^\dagger,b}.
\]
Applying RIO to the left, right, and middle blocks in turn gives
\[
\begin{aligned}
\mathsf S(\beta)
&\overset{d}{=}\bigl(\Obf(L),\Obf(L^\dagger|\beta|R),\Obf(R^\dagger)\bigr)\\
&\approx_c\bigl(\pi(D_L),\Obf(L^\dagger|\beta|R),\Obf(R^\dagger)\bigr)\\
&\approx_c\bigl(\pi(D_L),\Obf(L^\dagger|\beta|R),\pi(D_R)\bigr)\\
&\approx_c\bigl(\pi(D_L),\pi(D_0),\pi(D_R)\bigr)
\overset{d}{=}\mathsf I(\beta).
\end{aligned}
\]
Here $\overset{d}{=}$ denotes equality in distribution. Each comparison regards the other two blocks as the public view, which may contain ideal samples as permitted by RIO. All calls use fresh independent randomness. The middle-block replacement uses the conditioning argument above, and the final equality uses $M_1=g$ and $\Pi_1=\pi$.

For the joining step, let $u=|A|$, $v=|B|$. We claim that independent $X\gets\mathsf I(A)$ and $Y\gets\mathsf I(B)$ satisfy
\[
\mathsf{Join}(X,Y)\approx_c\mathsf I(A|B).
\]
Write $(L,S)$ for the masks of $X$ and $(T,R)$ for those of $Y$; all four are independent uniform $s$-gate circuits. The seam $X_R|Y_L$ implements the channel of $S^\dagger|T$. The seam clause of RIO replaces its obfuscation by $\pi(D_{\mathrm{seam}})$, regarding $(X_L,X_0,Y_0,Y_R)$ as the public view. After this replacement, the middle block is
\[
\Pi_u(D_A)|\pi(D_{\mathrm{seam}})|\Pi_v(D_B),
\]
where, conditional on the masks, the samples are independent and distributed as
\[
D_A\gets\calE_{L^\dagger|A|S,M_u},\qquad
D_{\mathrm{seam}}\gets\calE_{S^\dagger|T,e},\qquad
D_B\gets\calE_{T^\dagger|B|R,M_v}.
\]
The alternating block layout in the definition of $\Pi_m$ gives
\[
\Pi_u(D_A)|\pi(D_{\mathrm{seam}})|\Pi_v(D_B)
\overset{d}{=}
\Pi_{u+v}(D_A|D_{\mathrm{seam}}|D_B),
\]
with independent randomness on both sides. We can therefore apply SCP to the concatenation and then apply this common postprocessing:
\[
\begin{aligned}
D_A|D_{\mathrm{seam}}|D_B
&\approx_c
 \calE_{L^\dagger|A|T,M_u+e}
 \mid\calE_{T^\dagger|B|R,M_v}\\
&\approx_c
 \calE_{L^\dagger|A|B|R,M_{u+v}},
\end{aligned}
\]
where each $\calE$ denotes a sample from the indicated distribution, and the two samples on the first right-hand side are conditionally independent given $L,T,R$. The first comparison merges $D_A,D_{\mathrm{seam}}$ across the random separator $S$; the second merges the result with $D_B$ across $T$. Both regard $(L,R,X_L,Y_R)$ as side information, and the first also includes $T,D_B$. This side information is independent of the separator being removed. The SCP length conditions follow from $e\geq\mu(2s+1)$ and
\[
M_k\geq kg\geq\mu(k+2s+1)\quad(k\geq1),
\qquad M_u+e+M_v=M_{u+v}.
\]
Applying $\Pi_{u+v}$ turns the last distribution into the middle block of $\mathsf I(A|B)$, with the correct outer blocks $X_L,Y_R$ unchanged. This proves the joining claim.

The single-gate case and joining claim imply the invariant by induction: first replace the two child triples by their independent ideal versions, then apply the joining claim. These replacements remain valid for inefficiently samplable ideal triples, since the other independent triple can be fixed as nonuniform advice. There are $O(m)$ recursive calls, so the resulting polynomial-length hybrid has negligible total distinguishing advantage. For exactly equivalent, same-shape $C_0,C_1$, we now have
\[
\mathsf S(C_0)\approx_c\mathsf I(C_0)
\overset{d}{=}\mathsf I(C_1)\approx_c\mathsf S(C_1).
\]
Concatenating the triples proves the required indistinguishability.
\end{proof}

\paragraph{From exact equivalence to qiO.}
Exact equivalence makes the ideal distributions identical, whereas negligibly close but unequal channels have disjoint sets of exact implementations. To obtain qiO as defined in \Cref{def:qio}, we introduce an additional stability assumption on the postprocessed ideal distributions.

We first express the obfuscation using a single ideal sample. For $m\geq1$, set $L(\lambda,n,m):=2b+M_m$ and define
\[
\rho(D_L|D_0|D_R)
:=\pi(D_L)|\Pi_m(D_0)|\pi(D_R),
\]
where the input block lengths are $b,M_m,b$, respectively, and all calls use fresh independent randomness. Both $L$ and $\rho$ depend only on $(\lambda,n,m)$.

For the independent masks $T_L,T_R$ defining $\mathsf I(C)$, two applications of SCP give
\[
\begin{aligned}
&\calE_{T_L,b}
 \mid\calE_{T_L^\dagger|C|T_R,M_m}
 \mid\calE_{T_R^\dagger,b}\\
&\quad\approx_c
 \calE_{C|T_R,b+M_m}
 \mid\calE_{T_R^\dagger,b}\\
&\quad\approx_c
 \calE_{C,L}.
\end{aligned}
\]
Here each $\calE$ denotes a sample, with conditional independence given the masks. Applying the same $\rho$ throughout and using the invariant from the theorem's proof yields
\begin{equation}
\label{eq:ancilla-free-ideal-representation}
\operatorname{concat}(\mathsf S(C))
\approx_c\operatorname{concat}(\mathsf I(C))
\approx_c\rho(D),
\qquad D\gets\calE_{C,L(\lambda,n,m)}.
\end{equation}

\begin{assumption}[Stability of the ideal distributions]
\label{assumption:stability}
Let $C_0,C_1$ be any same-shape polynomial-size ancilla-free circuit families with common width $N=N(\lambda)\geq\lambda$, common gate count $m=m(\lambda)\geq1$, and negligible diamond distance between their channels. For $L=L(\lambda,N,m)$ and the corresponding postprocessor $\rho$ defined above,
\[
\rho(1^\lambda,D_0)\approx_c\rho(1^\lambda,D_1),
\qquad D_j\gets\calE_{C_j,L}\quad(j\in\{0,1\}),
\]
against nonuniform quantum polynomial-time distinguishers, with independent ideal samples and fresh independent randomness for the two calls to $\rho$.
\end{assumption}

The width restriction is part of the stability hypothesis. The lifting in \Cref{sec:tidy-qio-lifting} pads the compiled circuits to meet it before obfuscation.

\begin{corollary}[qiO for ancilla-free unitary circuits]
\label{cor:ancilla-free-qio}
Under \Cref{assumption:qSCP,assumption:qRIO,assumption:stability}, the local-mixing construction is a quantum-secure qiO in the sense of \Cref{def:qio} for polynomial-size ancilla-free unitary circuit families of width $N=N(\lambda)\geq\lambda$, with exact correctness.
\end{corollary}

\begin{proof}
Empty circuits are returned unchanged. For same-shape circuits $C_0,C_1$ of positive gate count whose channels are negligibly close, sample $D_j\gets\calE_{C_j,L}$ independently. Then
\[
\operatorname{concat}(\mathsf S(C_0))
\approx_c\rho(D_0)
\approx_c\rho(D_1)
\approx_c\operatorname{concat}(\mathsf S(C_1)).
\]
The outer comparisons follow from \eqref{eq:ancilla-free-ideal-representation}, and the middle comparison is \Cref{assumption:stability}.
\end{proof}

\subsection{Obfuscation for Tidy Circuits}
\label{sec:tidy-qio-lifting}

We prove that qiO for ancilla-free unitary circuits implies qiO for tidy unitary circuits. The key step of the proof is the circuit compilation discussed below (\Cref{fig:tidy-ancilla-free-compilation}), which turns closeness on the clean-input subspace into closeness on the entire input space. A random phase then lets us couple the compiled circuits even when the original clean-input actions agree only up to global phase.

Fix the logical and work registers $Q,A$, of sizes $n,a$, and write $P=I_Q\otimes\ketbra{0^a}_A$, and $\overline P=I-P$. Add a qubit $D$ and define, for any unitary $W$ on $Q,A$,
\begin{align}
F&=\sfX_D\otimes P+I_D\otimes\overline P,
&K(W)&=\ketbra{0}_D\otimes W+\ketbra{1}_D\otimes I,\notag\\
\Gamma(W)&=K(W)F K(W)^\dagger F.
\label{eq:tidy-gamma-definition}
\end{align}
Thus $F$ flips $D$ exactly when $A=0^a$, and $K(W)$ applies $W$ controlled on $D=0$. Products act from right to left.
\Cref{fig:tidy-ancilla-free-compilation} depicts $\Gamma_\alpha(C):=\Gamma(e^{i\alpha}V_C)$. All $n+a+1$ wires are logical inputs to this unitary. For a full-register unitary $W$, we write $\Phi_W(\rho)=W\rho W^\dagger$ for its channel.

\paragraph{The exactly tidy case.}
Let $J$ be the isometry $J\ket\psi=\ket\psi_Q\ket{0^a}_A$, and let $S=\operatorname{im}(J)$ be the clean-input subspace.
To see the purpose of the compilation, first suppose $VJ=JU$ for a logical unitary $U$, absorbing any global phase into $U$. Since $V$ is unitary, it preserves both $S$ and $S^\perp$. On $\mathbb C^2\otimes S$, $F$ flips $D$; on $\mathbb C^2\otimes S^\perp$, $F$ is the identity, so $K(V)$ and $K(V)^\dagger$ cancel. Identifying $S$ with the logical space via $J$, and writing blocks in the computational basis of $D$, we obtain
\[
\left.\Gamma(V)\right|_{\mathbb C^2\otimes S}
=\begin{pmatrix}U&0\\0&U^\dagger\end{pmatrix},
\qquad
\left.\Gamma(V)\right|_{\mathbb C^2\otimes S^\perp}=I.
\]
Thus the compiled unitary depends only on the clean-input action $U$; the original action on $S^\perp$ disappears.

There is still a phase ambiguity. Although $V$ and $e^{i\theta}V$ implement the same channel, their compilations need not: on the clean subspace,
\[
\left.\Gamma(e^{i\theta}V)\right|_{\mathbb C^2\otimes S}
=\begin{pmatrix}e^{i\theta}U&0\\0&e^{-i\theta}U^\dagger\end{pmatrix}.
\]
A global phase of $V$ therefore becomes a \emph{relative} phase between the two branches of $D$. The random phase and coupling argument below handle this ambiguity. The next lemma formalizes dependence on the clean-input action and bounds the effect of approximation errors, without assuming exact tidiness.

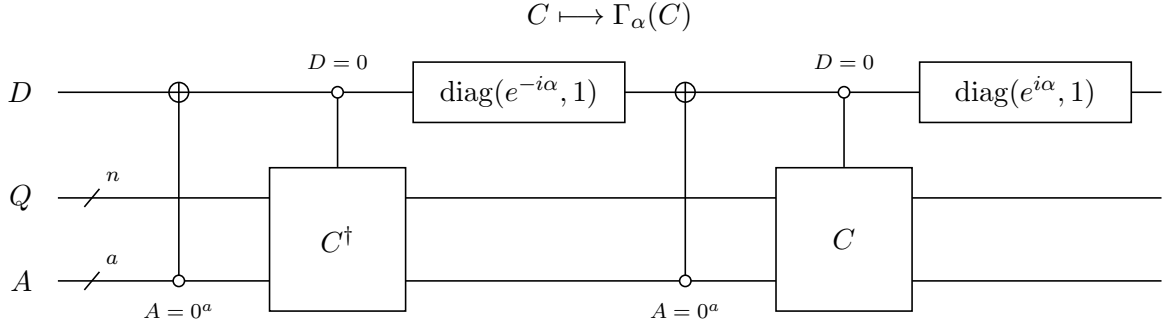
\begin{figure}[t]
    \centering
    \begin{tikzpicture}[x=1cm,y=1cm,line width=0.6pt,
        gate/.style={draw,fill=white,minimum width=1.8cm,
            minimum height=1.9cm,inner sep=5pt},
        phase/.style={draw,fill=white,minimum width=2.8cm,
            minimum height=0.8cm,inner xsep=7pt,inner ysep=5pt},
        control label/.style={font=\scriptsize,inner sep=2pt}]
        \node at (7.3,3.5) {$C\longmapsto\Gamma_\alpha(C)$};
        \node[anchor=east] at (-0.2,2.5) {$D$};
        \node[anchor=east] at (-0.2,1.1) {$Q$};
        \node[anchor=east] at (-0.2,0) {$A$};
        \foreach \y in {0,1.1,2.5} {
            \draw (0,\y) -- (14.6,\y);
        }
        \draw (0.35,0.98) -- (0.55,1.22);
        \draw (0.35,-0.12) -- (0.55,0.12);
        \node[anchor=south west,font=\scriptsize] at (0.5,1.18) {$n$};
        \node[anchor=south west,font=\scriptsize] at (0.5,0.08) {$a$};
        \node[gate] (inverse) at (3.7,0.55) {$C^\dagger$};
        \node[phase] at (6.1,2.5) {$\operatorname{diag}(e^{-i\alpha},1)$};
        \node[gate] (forward) at (10.4,0.55) {$C$};
        \node[phase] at (12.8,2.5) {$\operatorname{diag}(e^{i\alpha},1)$};
        \foreach \x in {1.6,8.3} {
            \draw (\x,0) -- (\x,2.5);
            \draw[fill=white] (\x,0) circle (0.075);
            \node[control label,anchor=north] at (\x,-0.22) {$A=0^a$};
            \draw[fill=white] (\x,2.5) circle (0.13);
            \draw (\x-0.13,2.5) -- (\x+0.13,2.5);
            \draw (\x,2.37) -- (\x,2.63);
        }
        \draw (3.7,2.5) -- (inverse.north);
        \draw (10.4,2.5) -- (forward.north);
        \foreach \x in {3.7,10.4} {
            \draw[fill=white] (\x,2.5) circle (0.075);
            \node[control label,anchor=south] at (\x,2.73) {$D=0$};
        }
    \end{tikzpicture}
    \caption{Ancilla-free compilation of a tidy circuit $C$. Sample
    $k$ once during obfuscation, uniformly from
    $\{0,\ldots,2^\lambda-1\}$, and set
    $\alpha=2\pi k/2^\lambda$. Gates act from left to right. An open
    control on $A$ tests whether the entire register is $0^a$; an open
    control on $D$ tests whether $D=0$. The boxes $C$ and $C^\dagger$
    denote the full gate sequence on $Q,A$ and its inverse, with no
    initialization or discarding inside either box. The diagonal gates
    $\operatorname{diag}(e^{-i\alpha},1)$ and
    $\operatorname{diag}(e^{i\alpha},1)$ act on $D$, multiplying its
    $\ket{0}$ component by the indicated phase and leaving its $\ket{1}$
    component unchanged. All $n+a+1$ qubits are logical inputs.
    After synthesis to negligible error and padding to a common shape,
    the classical description of this circuit is the input to
    $\Obf_{\ancfree}$.}
    \label{fig:tidy-ancilla-free-compilation}
\end{figure}

\begin{lemma}[Dependence on the clean-input restriction]
\label{lem:tidy-clean-restriction}
For any unitary $W$, the unitary $\Gamma(W)$ depends only on $WP$.
For any two unitaries $W_0,W_1$ on $Q,A$,
\begin{equation}
\|\Gamma(W_0)-\Gamma(W_1)\|
\leq 3\|(W_0-W_1)P\|.
\label{eq:tidy-gamma-lipschitz}
\end{equation}
Moreover, for every logical input $\ket\psi$,
\begin{equation}
\Gamma(W)(\ket0_D\otimes J\ket\psi)
=\ket0_D\otimes WJ\ket\psi.
\label{eq:tidy-gamma-evaluation}
\end{equation}
These statements do not require $W$ to be tidy.
\end{lemma}

\begin{proof}
Set $Z=WP$. In the computational basis of $D$, block multiplication gives
\begin{equation}
\Gamma(W)=
\begin{pmatrix}
I-ZZ^\dagger & Z\\
Z^\dagger & \overline P
\end{pmatrix}F,
\label{eq:tidy-gamma-blocks}
\end{equation}
which proves the claimed dependence. For $Z_b=W_bP$ and $\Delta=Z_0-Z_1$, we have $\|Z_b\|\leq1$ and
\[
\|Z_0Z_0^\dagger-Z_1Z_1^\dagger\|
=\|\Delta Z_0^\dagger+Z_1\Delta^\dagger\|
\leq2\|\Delta\|,
\qquad
\left\|\begin{pmatrix}0&\Delta\\\Delta^\dagger&0\end{pmatrix}\right\|
=\|\Delta\|.
\]
Since $F$ is unitary, the triangle inequality in \eqref{eq:tidy-gamma-blocks} gives \eqref{eq:tidy-gamma-lipschitz}.
Finally, on $\ket0_D\otimes J\ket\psi$, the first $F$ flips $D$ to $1$, $K(W)^\dagger$ acts trivially, the second $F$ restores $D$ to $0$, and $K(W)$ applies $W$. This proves \eqref{eq:tidy-gamma-evaluation}.
\end{proof}

\begin{theorem}[Ancilla-free qiO implies qiO for tidy circuits]
\label{thm:ancilla-free-to-tidy-qio}
Suppose $\Obf_{\ancfree}$ is a qiO in the sense of \Cref{def:qio} for all polynomial-size ancilla-free unitary circuits. Then there is a classical polynomial-time transformation constructing a qiO for any polynomial-size class of tidy unitary circuits whose tidiness error is bounded by a fixed negligible function. The same implication holds with classical distinguishers throughout.
\end{theorem}

\begin{proof}
We first construct the obfuscator using $\Gamma$ and prove its correctness. To prove security, we couple the random phases for circuits with negligibly close logical channels, show that their compiled channels are negligibly close on the full input space, and apply ancilla-free qiO security.

\paragraph{Construction.}
Set $M=2^\lambda$, $\alpha_k=2\pi k/M$, and $\tau=2^{-\lambda}$.
Given $(1^\lambda,C)$, sample $k$ uniformly from $\mathbb Z_M$ and compile $\Gamma(e^{i\alpha_k}V_C)$ into an ancilla-free circuit $G_{C,k}$ satisfying
\begin{equation}
\|\Phi_{G_{C,k}}-\Phi_{\Gamma(e^{i\alpha_k}V_C)}\|_\diamond
\leq\tau.
\label{eq:tidy-synthesis-error}
\end{equation}
This compilation has size and running time polynomial in $\lambda$ and $|C|$. The zero test $F$ has a polynomial-size decomposition without initialized ancillas \cite[Corollary~7.6]{Barenco1995Elementary}.
Controlling each gate of $C$ and $C^\dagger$ implements the remaining controlled operations. Their constituent gates and the sampled phase rotations can be synthesized over the fixed inverse-closed universal gate set with polynomial overhead, allocating sufficient accuracy to each gate to obtain \eqref{eq:tidy-synthesis-error}
\cite{DawsonNielsen2006SolovayKitaev}. Pad with identity gates so that the shape of $G_{C,k}$ depends only on $\lambda$ and $\shape(C)=(s,n,a,n)$, independently of $k$.
We may also pad $G_{C,k}$ to width $\max\{n+a+1,\lambda\}$ with additional identity wires, which preserves diamond distances. This allows us to use ancilla-free qiO for widths at least $\lambda$, as required by the stability assumption (\Cref{assumption:stability}).

Compute $\widehat G\gets\Obf_{\ancfree}(1^\lambda,G_{C,k})$ and return the circuit that, on an $n$-qubit input in $Q$, initializes $D,A$ to zero, applies $\widehat G$, and discards $D,A$.

\paragraph{Correctness.}
By \eqref{eq:tidy-gamma-evaluation}, for every $k$,
\[
\Gamma(e^{i\alpha_k}V_C)(\ket0_D\otimes J\ket\psi)
=e^{i\alpha_k}\ket0_D\otimes V_CJ\ket\psi.
\]
This identity also holds with an arbitrary reference register. Hence initializing $D,A$ to zero, applying $\Gamma(e^{i\alpha_k}V_C)$, and discarding $D,A$ induces exactly the channel implemented by $C$, including when $C$ is only approximately tidy. Let $\mu$ be the correctness error of $\Obf_{\ancfree}$. For each $k$, except with probability at most $\mu$, the channel of $\widehat G$ is within diamond distance $\mu$ of $\Phi_{G_{C,k}}$. Initialization and discarding cannot increase diamond distance, so \eqref{eq:tidy-synthesis-error} bounds the distance of the output circuit's channel from $C$ by $\mu+\tau$. The same success probability holds after averaging over $k$. Thus $\mu+\tau$ is a negligible correctness bound as required by \Cref{def:qio}.

\paragraph{Aligning the clean-input actions.}
Fix any two circuit families $C_0,C_1$ of common shape $(s,n,a,n)$ satisfying the hypotheses of qiO, suppressing their dependence on $\lambda$. Write $V_b=V_{C_b}$, let $\delta$ be a common negligible bound on their tidiness error, and suppose $\|C_0-C_1\|_\diamond\leq\varepsilon$ for negligible $\varepsilon$. Absorbing the tidiness phases into the logical unitaries, choose $U_b$ with $\|V_bJ-JU_b\|\leq\delta$ and set $\calU_b(\rho)=U_b\rho U_b^\dagger$. Closeness of these isometries and contractivity under discarding give
\[
\|C_b-\calU_b\|_\diamond\leq2\delta,
\qquad
\|\calU_0-\calU_1\|_\diamond\leq\varepsilon+4\delta.
\]
For unitary channels, there is a phase $\varphi$ such that\footnote{Fix an eigenvalue $\zeta_*$ of $U_0^\dagger U_1$. Testing an equal superposition of its eigenvector and an eigenvector with eigenvalue $\zeta_j$ gives $|\zeta_j-\zeta_*|\leq\|\calU_0-\calU_1\|_\diamond$ for every $j$. Choose $e^{i\varphi}=\overline{\zeta_*}$.}
\[
\|U_0-e^{i\varphi}U_1\|
\leq\|\calU_0-\calU_1\|_\diamond.
\]
Consequently,
\begin{equation}
\|(V_0-e^{i\varphi}V_1)P\|
=\|(V_0-e^{i\varphi}V_1)J\|
\leq\varepsilon+6\delta=:\eta.
\label{eq:tidy-phase-alignment}
\end{equation}
This is the only part of the proof that uses tidiness.

\paragraph{Coupling the phase samples.}
Choose $h\in\mathbb Z_M$ whose angle $\alpha_h$ is nearest to $\varphi$ modulo $2\pi$, so that $|e^{i\alpha_h}-e^{i\varphi}|\leq\pi/M$. Couple the two constructions by sampling $k$ uniformly and using indices $k$ for $C_0$ and $k+h$ for $C_1$, with indices taken modulo $M$. Both marginal indices are uniform, since the shift is a permutation of $\mathbb Z_M$. For every $k$, \eqref{eq:tidy-phase-alignment} gives
\[
\|(e^{i\alpha_k}V_0-e^{i\alpha_{k+h}}V_1)P\|
\leq\eta+\pi/M.
\]
Applying \Cref{lem:tidy-clean-restriction} yields
\[
\|\Gamma(e^{i\alpha_k}V_0)
-\Gamma(e^{i\alpha_{k+h}}V_1)\|
\leq3(\eta+\pi/M).
\]
For unitaries $A,B$, the channel bound
$\|\Phi_A-\Phi_B\|_\diamond\leq2\|A-B\|$, together with the two
synthesis errors, therefore implies
\begin{equation}
\|\Phi_{G_{C_0,k}}-\Phi_{G_{C_1,k+h}}\|_\diamond
\leq6(\eta+\pi/M)+2\tau=\negl(\lambda).
\label{eq:tidy-coupled-channels}
\end{equation}
The bound is uniform in $k$ and holds on the entire input space,
including coherent superpositions of clean and nonclean work states.
Thus the coupling preserves both sampling distributions and pairs
same-shape ancilla-free circuits with negligibly close channels.

\paragraph{From paired circuits to indistinguishable distributions.}
Fix a nonuniform quantum polynomial-time distinguisher $\calD$ and set
\[
p_{b,k}=\Pr\!\left[
\calD\bigl(1^\lambda,
\Obf_{\ancfree}(1^\lambda,G_{C_b,k})\bigr)=1
\right],
\]
where the probability is over the randomness of the obfuscator and distinguisher.
The advantage between the two distributions of obfuscated compiled circuits is
\begin{align}
\left|\frac1M\sum_k p_{0,k}-\frac1M\sum_\ell p_{1,\ell}\right|
&=\left|\frac1M\sum_k(p_{0,k}-p_{1,k+h})\right|\notag\\
&\leq\frac1M\sum_k|p_{0,k}-p_{1,k+h}|\notag\\
&\leq\max_k|p_{0,k}-p_{1,k+h}|.
\label{eq:tidy-coupling-average}
\end{align}
The maximum is negligible. Otherwise, for each $\lambda$ choose a
maximizing index $k_\lambda$. The families
$G_{C_0,k_\lambda}$ and $G_{C_1,k_\lambda+h_\lambda}$ have polynomial
size and the same shape, and their channels are negligibly close by
\eqref{eq:tidy-coupled-channels}, yet $\calD$ distinguishes their
obfuscations with nonnegligible advantage. This contradicts
\Cref{def:qio}, instantiated with the class of ancilla-free unitary
circuits, which allows circuit families that are not uniformly generated.

Finally, adding the fixed initialization and discarding operations to the
circuit descriptions is efficient classical postprocessing, so it preserves
this indistinguishability. This proves qiO security for tidy unitary circuits.
The argument also applies to classical distinguishers when the assumed
ancilla-free qiO has classical security.
\end{proof}

The identity \eqref{eq:tidy-gamma-evaluation} also shows that single-state
correctness of the underlying obfuscator suffices for single-state
correctness of the tidy-circuit construction: on the logical all-zero input,
the input to the obfuscated circuit is all zero as well.

The tidiness requirement is essential for the proof of \Cref{thm:ancilla-free-to-tidy-qio}. For example, $I_A$ and $\sfX_A$ both induce the identity logical channel after initializing and discarding a one-qubit work register $A$, but the latter does not return $A$ to zero. Their compiled circuits act differently on an input with $D=0,A=1$: the first leaves $D,A$ at $0,1$, while the second sends them to $1,0$. Thus the compilation does not give qiO for arbitrary circuits with ancillas or arbitrary quantum channels.

We now connect qiO for tidy unitary circuits to classical functionality and to the coherent-decision circuits used in our applications.

\begin{lemma}[Tidy circuits for classical functions]
\label{lem:tidy-classical-functionality}
Let $C$ have shape $(s,n,a,r)$ and output $f(x)$ with probability at least
$1-\varepsilon$ on every computational-basis input $x$, for a classical
function $f:\{0,1\}^n\to\{0,1\}^r$. There is a classical polynomial-time
transformation producing a circuit $T_C$ of shape
$(2s+2n+r,n+r,n+a,n+r)$ such that, writing $J=J_{n+r,n+a}$,
\[
\|V_{T_C}J-JU_f\|\leq2\sqrt{\varepsilon},
\qquad
U_f\ket{x}\ket{y}=\ket{x}\ket{y\oplus f(x)}.
\]
\end{lemma}

\begin{proof}
Keep the logical registers $X,Y$ and initialize fresh registers $R,A$
of sizes $n,a$ to zero. Copy $X$ into $R$ using $n$ CNOT gates, apply
the full unitary $V_C$ to $R,A$, XOR its designated output into $Y$,
apply $V_C^\dagger$, and undo the copy from $X$ to $R$.

For each fixed $x$, the state after $V_C$ has a component of norm at
most $\sqrt{\varepsilon}$ on which the designated output differs from
$f(x)$. The actual XOR and the ideal XOR by $f(x)$ agree on the
remaining component and differ in operator norm by at most $2$ on
this error component. Thus their difference has norm at most
$2\sqrt{\varepsilon}$ for every state of $Y$, including one entangled
with a reference. Uncomputing preserves this bound and, for the ideal
XOR, restores $R,A$ to zero. Since $X$ is preserved throughout,
the errors for distinct $x$ lie in orthogonal subspaces, so the same
bound holds on every superposition of inputs.
\end{proof}

In particular, if $\varepsilon=\negl(\lambda)$, the lemma gives a negligible tidiness bound $2\sqrt{\varepsilon}=\negl(\lambda)$. Thus obfuscating $T_C$ gives obfuscation of $U_f$: same-shape circuits computing the same classical function with negligible error have indistinguishable obfuscations of their lifts. This preserves the classical functionality in a coherent form; it need not preserve the original channel implemented by $C$.

\begin{corollary}[qiO for coherent-decision circuits]
\label{cor:tidy-qio-coherent-decision}
QiO for polynomial-size tidy unitary circuits implies qiO for
$\calC_{\BQP}$. The implication holds with either classical or quantum
nonuniform security.
\end{corollary}

\begin{proof}
Given $C\in\calC_{\BQP}$ of shape $(s,1,a,1)$, run its full unitary
on a fresh all-zero register of size $a+1$, CNOT its output into a
logical target qubit, and uncompute. Call the resulting circuit $T_C$;
its shape is $(2s+1,1,a+1,1)$.

For any polynomial-size family $\{C_\lambda\}_{\lambda\in\NN}$ in
$\calC_{\BQP}$, let $b_{C_\lambda}$ denote its semantic bit. By the
definition of $\calC_{\BQP}$, there is a negligible function
$\varepsilon$ such that
\[
\|C_\lambda-\sfX^{b_{C_\lambda}}\|_\diamond
\leq\varepsilon(\lambda).
\]
In particular, evaluating $C_\lambda$ on $\ket0$ outputs
$b_{C_\lambda}$ with error at most $\varepsilon(\lambda)$. Applying
\Cref{lem:tidy-classical-functionality} with no variable input gives
\[
\|V_{T_{C_\lambda}}J-J\sfX^{b_{C_\lambda}}\|
\leq2\sqrt{\varepsilon(\lambda)},
\qquad J=J_{1,a+1}.
\]
Consequently,
\[
\|T_{C_\lambda}-C_\lambda\|_\diamond
\leq4\sqrt{\varepsilon(\lambda)}+\varepsilon(\lambda)=:\kappa(\lambda),
\]
where $\kappa$ is negligible. Return the qiO obfuscation of $T_{C_\lambda}$. Correctness follows by adding $\kappa$ to the tidy obfuscator's correctness error. For any same-shape pair $C_{0,\lambda},C_{1,\lambda}$ with channel distance at most a negligible $\varepsilon'$, the lifts have the same shape and channel distance at most $\varepsilon'+2\kappa$. QiO security for tidy unitary circuits therefore gives the required indistinguishability.
\end{proof}

To apply this construction to public verification, we also record the
classical auxiliary-input guarantee implied by nonuniform security.

\begin{lemma}[Classical auxiliary input]
\label{lem:qio-classical-auxiliary-input}
Suppose qiO security holds against nonuniform distinguishers for
arbitrary, not necessarily uniformly generated circuit families.
Let $(z,C_0,C_1)$ be any polynomial-size classical joint ensemble such
that every pair $C_0,C_1$ lies in the input class, has the same shape,
and has channel distance at most a common negligible function
$\varepsilon(\lambda)$. Then, with fresh obfuscation randomness,
\[
(1^\lambda,z,\Obf(1^\lambda,C_0))
\approx_c
(1^\lambda,z,\Obf(1^\lambda,C_1)).
\]
The auxiliary-input guarantee has the same classical or quantum security
as the assumed qiO.
\end{lemma}

\begin{proof}
Fix a distinguisher. Its advantage on the ensemble is at most the
maximum absolute conditional advantage over triples $(z,C_0,C_1)$
in the support. If this maximum were nonnegligible, choose a maximizing
triple for each $\lambda$ and hardwire its $z$ into the distinguisher.
The selected circuits form polynomial-size, same-shape families with
negligible channel distance, contradicting nonuniform qiO security.
\end{proof}

Combining \Cref{cor:ancilla-free-qio},
\Cref{thm:ancilla-free-to-tidy-qio}, and \Cref{cor:tidy-qio-coherent-decision}
shows that \Cref{assumption:qSCP,assumption:qRIO,assumption:stability}
imply quantum-secure qiO for $\calC_{\BQP}$.
By \Cref{lem:qio-classical-auxiliary-input}, this qiO also provides security
in the presence of correlated classical auxiliary input. Thus these assumptions supply the qiO
needed for private verification (\Cref{thm:private-bqp-verification}),
for publicly verifiable proofs of quantumness when
$\BQP\not\subseteq\mathrm{i.o.}\BPP$ (\Cref{thm:public-poq}), and for
public verification when one-way functions secure against nonuniform
quantum adversaries also exist
(\Cref{thm:public-bqp-verification}).

\section{Acknowledgements and AI disclosure}
We thank Tomoyuki Morimae for pointing out the statement of \Cref{lem:qio-implies-classical-owf}. We thank Petar Jurcevic for organizing the workshop ``Quantum advantage with efficient classical verification'' and the workshop participants for helpful discussions. Aparna Gupte thanks Vinod Vaikuntanathan and Seyoon Ragavan for useful discussions.

The authors proposed all the new definitions, constructions, and proof ideas. We used LLMs, including Claude Fable 5, Opus 5.5, GPT 5.6 Sol and GPT 6 Astra to assist with brainstorming, editing the text, and verifying the results. The authors take responsibility for all results in the paper.

\bibliographystyle{alpha}
\bibliography{refs.bib}

\end{document}